%% file: main.tex
\documentclass[11pt, letterpaper]{article}
\usepackage[utf8]{inputenc}
\usepackage{fullpage,times}

\usepackage{fullpage}

\usepackage{amsmath,amsfonts,amsthm,amssymb,multirow}
\usepackage{times}
\usepackage{graphicx}
\usepackage{floatpag}
\usepackage{algorithm}
\usepackage[noend]{algpseudocode}
\usepackage{bbold}
\usepackage{enumitem}
\usepackage{subcaption} 
\usepackage{calc}
\usepackage{tikz}
\usetikzlibrary{decorations.markings}
\tikzstyle{vertex}=[circle, draw, inner sep=0pt, minimum size=4pt, fill = black]

\usepackage{graphicx}
\usepackage{tabularx}
\makeatletter
\newcommand{\multiline}[1]{%
  \begin{tabularx}{\dimexpr\linewidth-\ALG@thistlm}[t]{@{}X@{}}
    #1
  \end{tabularx}
}
\makeatother
\usepackage{mathtools}

\makeatletter
\def\BState{\State\hskip-\ALG@thistlm}
\makeatother

\usepackage[compact]{titlesec}
\titlespacing{\section}{0pt}{3ex}{2ex}
\titlespacing{\subsection}{0pt}{2ex}{1ex}
\titlespacing{\subsubsection}{0pt}{0.5ex}{0ex}

\newtheorem{theorem}{Theorem}[section]

\newtheorem{corollary}{Corollary}[section]

\newtheorem{definition}{Definition}[section]
\newtheorem{lemma}{Lemma}[section]
\newtheorem{claim}{Claim}
\newtheorem{hypothesis}{Hypothesis}

\makeatletter
\let\c@fconjecture\c@conjecture
\makeatother

\makeatletter
\let\c@fconj\c@conj
\makeatother

\def \R {{\mathbb R}}

\def \eps {\varepsilon}

\newcommand{\ignore}[1]{}

\def \polylog { \text{\rm polylog~} }

\title{Truly Subcubic Min-Plus Product for Less Structured Matrices, with Applications}

\author{Virginia Vassilevska Williams\thanks{MIT EECS and CSAIL, virgi@mit.edu} ~and Yinzhan Xu\thanks{MIT EECS, xyzhan@mit.edu}}
\date{}

\begin{document}

\maketitle

\begin{abstract}
The All-Pairs Shortest Paths (APSP) problem is one of the most basic problems in computer science. The fastest known algorithms for APSP in $n$-node graphs run in $n^{3-o(1)}$ time, and it is a big open problem whether a truly subcubic, $O(n^{3-\eps})$ for $\eps>0$ time algorithm exists for APSP.
The Min-Plus product of two $n\times n$ matrices is known to be equivalent to APSP, where the optimal running times of the two problems differ by at most a constant factor. A natural way to approach understanding the complexity of APSP is thus understanding what structure (if any) is needed to solve Min-Plus product in truly subcubic time. The goal of this paper is to get truly subcubic algorithms for Min-Plus product for less structured inputs than what was previously known, and to apply them to versions of APSP and other problems. The results are as follows:

(1) Our main result is the first truly subcubic algorithm for the Min-Plus product of two $n\times n$ matrices $A$ and $B$ with $\polylog n$ bit integer entries, where $B$ has a partitioning into $n^{\eps}\times n^{\eps}$ blocks (for any $\eps>0$) where each block is at most $n^\delta$-far (for $\delta<3-\omega$, where $2\leq \omega<2.373$) in $\ell_\infty$ norm from a constant rank integer matrix.
This result presents the most general case to date of Min-Plus product that is still solvable in truly subcubic time.

(2) The first application of our main result is a truly subcubic algorithm for APSP in a new type of geometric graph.
Chan'10 solved APSP in truly subcubic time in geometric graphs whose edges have weights that are a function of the identities of the edge's end-points. Our result extends Chan's result in the case of integer edge weights by allowing the weights to differ from a function of the end-point identities by at most $n^\delta$ for small $\delta$.

(3) In a second application we consider a batch version of the range mode problem in which one is given a sequence of numbers $a_1,\ldots,a_n$ and $n$ intervals defining contiguous subsequences, and one is asked to compute the range mode of each subsequence. Chan et al.'14 showed that any $O(n^{1.5-\eps})$ time combinatorial algorithm for $\eps>0$ for this problem can be used to solve Boolean matrix multiplication combinatorially in truly subcubic time. We give the first $O(n^{1.5-\eps})$ time for $\eps>0$ algorithm for this batch range mode problem, showing that the hardness is indeed constrained to combinatorial algorithms.

(4) Our final application is to the
 Maximum Subarray problem: given an $n\times n$ integer matrix, find the contiguous subarray of maximum entry sum. We show that Maximum Subarray can be solved in  truly subcubic, $O(n^{3-\eps})$ (for $\eps>0$) time, as long as every entry of the input matrix is no larger than $O(n^{0.62})$ in absolute value. This is the first truly subcubic algorithm for an interesting case of Maximum Subarray. The Maximum Subarray problem with arbitrary integer entries is known to be subcubically equivalent to APSP, in that a truly subcubic, $O(n^{3-\eps})$ time algorithm for $\eps>0$ for one problem would imply a truly subcubic algorithm for the other. Because of this it is believed that Maximum Subarray does not admit truly subcubic algorithms, without a restriction on the inputs. 

We also improve all the known conditional hardness results for the $d$-dimensional variant of Maximum Subarray, showing that many of the known algorithms are likely tight.
\end{abstract}

\newpage{}
\section{Introduction}
\input{intro.tex}

\section{Preliminaries}
\input{preliminary.tex}

\section{Main Algorithm}\label{sec:upper_bound}
\input{upper_bound.tex}

\section{Application I: Geometric APSP}

\input{apsp}

\section{Application II: Batch Range Mode}
\input{range_mode.tex}

\section{Application III: Maximum Subarray with Bounded Entries}

\input{subarray}

\subsection{Tight Lower Bound for $d$-Dimensional Maximum Subarray}\label{sec:lower_bound}
\input{lower_bound.tex}

\subsection{Lower Bound for Maximum Subarray with Bounded Weight}\label{sec:lower_bound_bounded}
\input{lower_bound_bounded.tex}

\bibliographystyle{plain}
\bibliography{ref}

\appendix
\input{appendix.tex}

\end{document}

%% file: intro.tex

The All-Pairs Shortest Paths (APSP) problem is one of the most basic and well-studied problems in graph algorithms. Algorithms for APSP have been studied since the 1950s when the Floyd-Warshall algorithm achieved a running time of $O(n^3)$ in $n$-vertex graphs. Over the next six decades, some improvements over the cubic running time were developed, culminating in the current fastest $n^3/2^{\Theta(\sqrt{\log n})}$ time algorithm by Williams~\cite{williams2014faster}. Unfortunately, no {\em truly subcubic}, $O(n^{3-\eps})$ time for $\eps>0$ algorithm is known, and a hypothesis that such an algorithm does not exist for APSP has become prominent in the field of fine-grained complexity (see e.g. \cite{vsurvey}).

The so called {\em Min-Plus} product of matrices $A$ and $B$, defined as the matrix $C$ with $C_{i,j}=\min_k (A_{i,k}+B_{k,j})$, is known to be asymptotically {\em equivalent} to APSP (see e.g. \cite{fischermeyer}) in the sense that a $T(n)$ time algorithm for the Min-Plus product of two $n\times n$ matrices implies an $O(T(n))$ time algorithm for APSP in $n$-node graphs, and vice versa. Because of this equivalence, research on APSP algorithms typically involves studying the Min-Plus product directly.

A long line of research on APSP involves studying the Min-Plus product of structured matrices. The relationship between APSP and Min-Plus product extends to structured instances as well: Min-Plus product of structured matrices can be viewed as APSP in a structured layered graph with three layers.
Conversely, in the case of graphs with integer weights but also in many other cases, APSP on structured instances  is in truly subcubic time 
if Min-Plus product of an {\em arbitrary} matrix with a {\em structured} matrix (i.e. the graph's generalized adjacency matrix) is in truly subcubic time\footnote{For graphs with integer edge weights, this is true regardless of the structure, as one can always leverage two types of approaches to APSP: (1) compute the distances on paths that have at most $n^\delta$ vertices by iterating the Min-Plus product $n^\delta$ times, and (2) compute the distances on paths that have at least $n^\delta$ vertices by random sampling and a SSSP algorithm such as Dijkstra's, after the usual removal of negative edge weights using Johnson's trick and a truly subcubic time SSSP algorithm such as \cite{Goldberg95}.}. 

Studying structured instances of Min-Plus product/APSP is important for two main reasons:
\begin{itemize}
\item As an approach towards truly subcubic APSP: {\em What structure, if any, is needed to solve APSP in truly subcubic time? }
\item As an approach to solve other problems faster: APSP is a very versatile problem, and many other important problems that sometimes, on the face of it, seem to have nothing to do with shortest paths, can be reduced to APSP. Often, the instances that are produced in these reductions are actually structured, and one could exploit this structure to get faster algorithms.
\end{itemize}

In the 1990s, Alon, Galil and Margalit~\cite{AlonGM97} showed that the Min-Plus product of two $n\times n$ matrices of integers in $\{M,\ldots,M\}$ can be computed in $O(Mn^\omega\log(Mn))$ time, where $\omega<2.373$ \cite{vmmult,legallmult} is the matrix multiplication exponent; thus Min-Plus product is in truly subcubic time, as long as the matrix entries are small, $M<O(n^{3-\omega-\eps})$ for $\eps>0$. This result is used over and over in shortest paths algorithms. For instance, it implies that APSP in undirected \cite{ShoshanZ99} and directed \cite{zwickbridge} graphs with small enough integer weights is in truly subcubic time. 

Truly subcubic time algorithms for less and less structured versions of Min-Plus and APSP were developed over the years, e.g. \cite{AlonGM97,Yuster09,chan2010more,vwstoc06}. The most general structured Min-Plus algorithm to date is by Bringmann et al. \cite{BringmannGSW16}: Min-Plus product of two $n\times n$ integer matrices $A$ and $B$ is in truly subcubic time if $A$ is arbitrary and for all rows (or similarly, columns) of $B$, any two consecutive entries are close: $|B[i,j]-B[i,j+1]|\leq n^\delta$ for small enough $\delta>0$. ($B$ is then called a bounded difference matrix.) 

Bringmann et al. showed that their result on bounded difference matrices subsumes all previous results on truly subcubic Min-Plus product. They also gave several applications of their new algorithm, most notably for language edit distance and RNA folding, that were not possible with the prior results on structured Min-Plus product.

Even though it is very powerful, the Bringmann et al. Min-Plus product result is still not general enough to solve some well-structured Min-Plus instances. We give one simple example. Consider a matrix $X$ such that for every $i,j$, $|X[i,j]+X[i+1,j+1]-X[i+1,j]-X[i,j+1]|\leq 1$; let's call $X$ a bounded discrete derivative (BDD) matrix. BDD matrices are extremely special, and we won't be too surprised if their Min-Plus product can be done in truly subcubic time. A truly subcubic algorithm for Min-Plus product for BDD matrices would be useful, for instance, for finding a Maximum Subarray of a matrix with small entries, a well-studied problem with many applications.

Unfortunately, however, BDD matrices are not bounded difference matrices, and the Bringmann et al. algorithm does not apply to them. Even the general framework devised by Bringmann et al. cannot be used as is. (We will go into more details in a bit.)  
The main goal of this paper is to modify Bringmann et al.'s framework to support less structured matrices, and to apply the new framework to obtain the first substantial improvements on the complexity of several studied problems.

\subsection{Our results}

\subsubsection{New Subcubic Min-Plus Products.}
Our main result is a new algorithm for Min-Plus product for less structured matrices. We begin with defining the structure needed.

\begin{definition}[$W$-approximate rank]
For an $n \times n$ integer matrix $M$, its $W$-approximate rank is defined as $$\min\left\{\text{rank}(X): X \in \mathbb{Z}^{n \times n}, |X-M|_{\infty} \le W \right\}.$$
\end{definition}
This $W$-approximate rank definition resembles the $\eps$-approximate rank definition of Alon et al.~\cite{alonapproxrank}. The difference is that we require the matrix $X$ to be have integer entries. 

Let $\delta>0$ be a constant and let $W\geq 0$ be an integer.
Consider an $n\times n$ integer matrix $B$ with the following structure. First partition $B$ into $n^\delta\times n^\delta$ blocks $B^{a,b}$ (containing the entries $B_{i,j}$ where $i\in (an^\delta,(a+1)n^{\delta}], j\in (bn^\delta,(b+1)n^{\delta}]$).
We require that every block submatrix $B^{a,b}$ has $W$-approximate rank at most $O(1)$. 

Our main result is:
\begin{theorem}\label{thm:low_rank}
Let $\delta\in (0,1]$. The Min-Plus product of two $n\times n$ matrices $A$ and $B$ whose entries are $\polylog n$ bit integers, and $B$ has all its $n^\delta\times n^\delta$ blocks of $W$-approximate rank at most $d$ for $1\leq d=O(1)$ can be computed in time 
\[\tilde{O}\left(n^{3-\frac{\delta}{\lfloor (d+1)/2 \rfloor}} + W^{1/4} n^{(9+\omega)/4}\right).\]
\end{theorem}

Notice that the matrix $A$ is arbitrary, as long as its entries do not get too huge, larger than $2^{\omega(\polylog n)}$. We would like arithmetic operations on the matrix entries to take $\tilde{O}(1)$ time, so that this entry size is not much of a restriction. The algorithm can handle larger entries as well. If the entries of $A$ and $B$ are $\beta$-bit integers, the algorithm gets a $\tilde{O}(\beta)$ overhead.

The running time of the algorithm is truly subcubic for any constant $d$ and any constant $\delta>0$, as long as $W=O(n^{3-\omega-\eps})$ for some $\eps>0$.

Let us discuss first why Theorem \ref{thm:low_rank} subsumes all previous results on truly subcubic structured Min-Plus product. We only need to show that a bounded differences matrix also has constant $W$-approximate rank blocks, as by the discussion in \cite{BringmannGSW16}, all other known cases of truly subcubic Min-Plus can be reduced to multiplying a bounded differences matrix with an arbitrary integer matrix. Suppose that $B$ is such that for every $i$ and $j$, $|B_{i,j}-B_{i,j+1}|\leq Q$ for small $Q$. Now consider the $n^\delta\times n^\delta$ sub-blocks $B^{a,b}$ of $B$ (for any choice of $\delta>0$). All columns of $B^{a,b}$ differ entrywise from the first column $B^{a,b}(1)$ by at most $Qn^\delta$. Thus, if we consider the rank one matrix that has $n^\delta$ columns identical to $B^{a,b}(1)$, we see that $B^{a,b}$ has $Qn^\delta$-approximate rank one. Hence by Theorem~\ref{thm:low_rank}, we get that for any $Q=O(n^{3-\omega-\eps})$ for $\eps>0$, we can pick $\delta=\eps/2$ and we'll get a truly subcubic time algorithm to Min-Plus multiply an arbitrary integer matrix $A$ by $B$.

Theorem~\ref{thm:low_rank} is very general and can handle much more than just bounded difference matrices. For instance, it is not hard to see that the aforementioned BDD matrices have constant $W$-approximate rank blocks, but also many other structured instances can be solved using Theorem~\ref{thm:low_rank}, as we will see in our applications. 

To prove Theorem~\ref{thm:low_rank} we modify the Min-Plus framework of Bringmann et al. \cite{BringmannGSW16} and combine it with a result from computational geometry on halfspace intersection reporting. 

We will give a brief overview on how we modify the Bringmann et al. framework.
The framework from \cite{BringmannGSW16} for computing the Min-Plus product $C$ of integer matrices $A$ and $B$ consists of Phase $1$, Phase $2$ and Phase $3$. 

Phase $1$ computes a matrix $C'$ which is close in $\ell_\infty$ norm to the desired output $C$. This phase is not hard to perform for the type of matrices we are considering; also, as shown by Bringmann et al., often this Phase can be avoided by scaling, and the real difficulty is in Phase 2, and especially Phase 3.

Phase $2$ iteratively takes random samples of rows of $A$ and columns of $B$, and repeatedly creates new matrices $\tilde{A}$ and $\tilde{B}$ whose entries are clever linear combinations of entries of $A,B$, the sampled row and column, and $C'$, so that most entries of the Min-Plus product $C$ of $A$ and $B$ can be easily computed from the Min-Plus products $\tilde{C}$ of $\tilde{A}$ and $\tilde{B}$ in $O(n^2)$ time. To perform Phase $2$ efficiently, Bringmann et al. replace any entries of $\tilde{A}$ and $\tilde{B}$ that are larger than some $M$ by $\infty$ and use the $\tilde{O}(Mn^\omega)$ time algorithm \cite{AlonGM97} to perform the Phase 2 Min-Plus products. By removing the large entries, some entries of $C$ will not be recoverable from the computed Min-Plus products $\tilde{C}$ in the Phase 2 iterations. Bringmann et al. show that at most a truly subcubic number of sums $A_{i,k}+B_{k,j}$ that might be close to the Min-Plus product entries will be missed in the computation.

Phase $3$ strives to recover the parts of the output matrix $C$ that are missed by Phase $2$. We know that at most a truly subcubic number of relevant sums $A_{i,k}+B_{k,j}$ need to be considered. If we knew which triples $i,k,j$ are involved in such sums, then we could finish the Min-Plus product in truly subcubic time by computing the sums explicitly. However, the main difficulty lies exactly in finding these triples. In particular, in the case of BDD matrices, there doesn't seem to be enough structure for one to be able to recover the remaining relevant triples in Step $3$ efficiently.

One of the main insights in this work is that one can offload more work to Phase 2 so that in Phase 3 there is enough structure left to recover the remaining relevant triples efficiently. In particular, instead of removing the large entries from both $\tilde{A}$ and $\tilde{B}$ in Phase 2, we only remove them from $\tilde{A}$. Then intuitively, Phase 2 does more work, and it turns out that in Phase 3, in truly subcubic time, one can find the remaining triples that one needs to consider to compute the entire Min-Plus product of $A$ and $B$, using a halfspace intersection reporting data structure from computational geometry.

However, now in Phase 2 we need to compute the Min-Plus product of an {\em arbitrary} integer matrix with a matrix with $\infty$ entries and finite entries bounded by $M$. This type of Min-Plus product is no longer known to be in $\tilde{O}(Mn^\omega)$ time. An $\tilde{O}(Mn^{(3+\omega)/2})$ time algorithm follows from prior work (e.g. \cite{Yuster09}, Lemma 3.3). We are able to improve the dependence on $M$, thus allowing for a faster truly subcubic final algorithm for Theorem~\ref{thm:low_rank}.

\begin{theorem}\label{thm:small_weight_times_large_weight}
The $(\min,+)$-product of two $n\times n$ integer matrices $A$ and $B$, where $A$ has entries in $\{-M,\ldots, M\}\cup\{\infty\}$ for some $M \ge 1$ and $B$ is arbitrary can be computed in $\tilde{O}(\sqrt{M}n^{(3+\omega)/2})$ time.
\end{theorem}

\subsubsection{Applications.}
To highlight the power of our new Min-Plus algorithm, we apply Theorem~\ref{thm:low_rank} to obtain the first improvements in the running times for several problems: a new geometric version of APSP, a batch range mode problem considered by Chan et al.~\cite{chan2014linear} and the Maximum Subarray problem. 

\paragraph{Geometric APSP.} As we discussed earlier, typically, an algorithm for a structured version of Min-Plus product implies an algorithm for a structured version of APSP. An almost immediate consequence of Theorem~\ref{thm:low_rank} is that APSP for graphs whose generalized adjacency matrix has $n^\delta\times n^\delta$ blocks of constant $W$-approximate rank and whose entries are $\polylog n$ bit integers can be solved in truly subcubic time when $\delta>0$ and $W\leq O(n^{3-\omega-\eps})$ for some $\eps>0$.

The proof is fairly standard: iterate the Min-Plus product of Theorem~\ref{thm:low_rank} $L$ times, where in the $i$th iteration $B$ is the generalized adjacency matrix of the graph and $A$ is the matrix computed in the $(i-1)$-th iteration (in the first iteration $A=B$). Then in the $L$th iteration one has computed the shortest paths in the graph using at most $L$ edges. To handle the paths longer than $L$ one computes SSSP from a random sample of $\tilde{O}(n/L)$ vertices, and $L$ is chosen to balance the running times.

Let us discuss what the graphs that we can handle look like: 
Define a $(W,d,\delta)$-geometrically weighted clustered graph, $(W,d,\delta)$-GWC for short as follows. $G=(V, E)$ is $(W,d,\delta)$-GWC if
\begin{itemize}
\item $V$ is partitioned into $t = n^{1-\delta}$ subsets $V_1, V_2, \ldots, V_t$ of size $O(n^\delta)$,
\item for every $i,j\in \{1,\ldots,t\}$, each $v \in V_i$ is assigned a $d$-dimensional integer vector $p^{i,j}(v)$, and each $u \in V_j$ is assigned a $d$-dimensional integer vector $q^{i,j}(u)$, and
\item for $v \in V_i$,  $u \in V_j$, $|w(v,u)-p^{i,j}(v)^T q^{i,j}(u)|\leq W.$ In other words, the edge weights in $V_i\times V_j$ are determined by a matrix whose $W$-approximate rank is at most $d$.
\end{itemize}

Notice that $(W,d,\delta)$-GWC graphs can simulate a lot of structure. For instance, imagine that each vertex $j$ is represented by an integer $x_j$, and the weights are determined by some degree $d$ (for $d=O(1)$) polynomial function $p$ of $x_i$ and $x_j$, up to an error at most $W$. Then, the weights can be represented (up to noise at most $W$ in each entry) with inner products of vectors $v_i$ and $v'_j$ of length $d^2$, where $v_i[a,b]$ is the monomial of $p(x'_i,x'_j)$ corresponding to $(x_i')^{a}\cdot (x_j')^{b}$ with the corresponding coefficient coming from $p$, evaluated at $x'_i=x_i$ and $x'_j=1$, and $v'_j[a,b]$ is $x_j^b$;
then we get that $v_i^T v'_j=p(x_i,x_j)$. A similar argument can be carried over if the $x_i$ are $O(1)$ dimensional vectors and $p$ is a polynomial in the entries of $x_i$ and $x_j$.

In \cite{chan2010more}, Chan studied a related version of geometrically weighted APSP where the weights between two vertices can be arbitrary algebraic functions, instead of just dot products between two vectors or polynomials. We remark that if we replace the geometric data structure that our Theorem~\ref{thm:low_rank} uses (Theorem \ref{thm:geo}) with the partition theorem in \cite{agarwal1994range}, we can achieve APSP for arbitrary algebraic functions as in \cite{chan2010more}, as long as the produced edge weights are integers. 
Moreover, our algorithm allows the edge weights to disagree with the function of their endpoints by an additive error, while the algorithm in \cite{chan2010more} requires the edge weights to exactly agree with the function.
In other words, in the case of integer edge weights, we obtain a more powerful geometric APSP algorithm.

\paragraph{Batch Range Mode.}
Given a sequence $a$ of length $n$, the range mode query on a range $[l,r]$ asks for the frequency of the most frequent element in the subsequence between the $l$-th and $r$-th element of $a$. 
Chan et al. \cite{chan2014linear} designed a linear space data structure that answers any range mode query in $\tilde{O}(\sqrt{n})$ time. Because the preprocessing step of the data structure is fast, this implies a $\tilde{O}(n^{1.5})$ time algorithm for the {\em batch} range mode problem in which one is given a sequence and $n$ range mode queries to answer in batch. 

Chan et al. \cite{chan2014linear} showed that 
any
combinatorial algorithm for the batch range mode problem running in $O(n^{1.5-\eps})$ time for $\eps>0$ would imply an $O(n^{3-\delta})$ time  combinatorial algorithm for $\delta>0$ that computes the product of two $n$ by $n$ Boolean matrices. This suggests that it might be hard to find such a combinatorial algorithm for batch range mode, as Boolean matrix multiplication is often conjectured to require $n^{3-o(1)}$ time using a combinatorial algorithm (see e.g. \cite{vsurvey}). This leads to a natural question: if we do not limit to combinatorial algorithms, what should the complexity of batch range mode be? Prior to this work, no noncombinatorial $n^{1.5-o(1)}$ lower bounds (even conditional ones), and no $O(n^{1.5-\eps})$ time (for $\eps>0$) algorithms were known to exist. 

As another application of Theorem~\ref{thm:low_rank} we obtain a $\tilde{O}(n^{1.4854})$ time algorithm for batch range mode, giving the first ever $O(n^{1.5-\eps})$ time (for $\eps>0$) algorithm for the problem. Note that in this application, we use $d=1$ in Theorem~\ref{thm:low_rank}, so each block of matrix $B$ is a bounded difference matrix. 
Thus Bringmann et al.'s algorithm suffices to give an $O(n^{1.5-\eps})$ (for $\eps > 0$) time algorithm for batch range mode.

\paragraph{Maximum Subarray.}
In the Maximum Subarray problem, one is given a real valued square matrix and is asked to find the contiguous submatrix of maximum entry sum. First studied by Bentley~\cite{bentley84alg}, the problem has many applications, for instance in graphics (see \cite{tamaki1998algorithms}) and in databases \cite{databases-app1,databases-app2,databases-app2j,databases-app3,databases-app4}.

The Maximum Subarray problem can be generalized to arbitrary dimension $d$: here one is given a $d$-dimensional grid (or tensor) with $n$ coordinates in each dimension (i.e. $[n]^d$), each point in the grid has a real value, and one is asked to return the contiguous subgrid of maximum entry sum.
In 1D, Kadane's algorithm (presented in \cite{bentley84alg}) achieves a linear, $O(n)$ running time. Bentley~\cite{bentley84} showed how to use Kadane's algorithm to solve the $2$D variant of the Maximum Subarray problem in $O(n^3)$ time; the same approach gives an $O(n^{2d-1})$ time algorithm, ``Kadane's algorithm'', for the $d$ dimensional version for all $d$. Tamaki et al.~\cite{tamaki1998algorithms} and Takaoka~\cite{Takaoka02} showed how to use divide-and-conquer to efficiently reduce the $2$D Maximum Subarray problem on an $n\times n$ grid to the Min-Plus product of two $n\times n$ matrices. Using the fastest APSP algorithm to date by Williams~\cite{williams2014faster}, one obtains the fastest 2D Maximum Subarray algorithm to date, running in $n^3/2^{\Theta(\sqrt{\log n})}$ time. This algorithm can be used to give the fastest known running time  $n^{2d-1}/2^{\Theta(\sqrt{\log n})}$ for the $d$-dimensional version of the problem as well.

In recent years, fine-grained complexity has yielded conditional lower bounds for Maximum Subarray. Backurs et al.~\cite{backurs2016tight} and Vassilevska W. and Williams~\cite{vw10j} showed that an $O(n^{3-\eps})$ time algorithm for $2$D Maximum Subarray for $\eps>0$ would imply an $O(n^{3-\eps'})$ time algorithm for Min-Plus product (and hence APSP), for $\eps'>0$. Together with the reductions of \cite{tamaki1998algorithms,Takaoka02}, this implies that the $2$D Maximum Subarray problem is subcubically equivalent to APSP. One of the main hardness hypotheses of fine-grained complexity is that APSP requires $n^{3-o(1)}$ time in graphs with integer weights (in the word RAM model with $O(\log n)$ bit words). Under this hypothesis, the best known algorithms for $2$D Maximum Subarray are essentially optimal, up to $n^{o(1)}$ factors, for arbitrary integer matrices.

An intriguing question is whether the $2$D Maximum Subarray problem can be solved in truly subcubic, $O(n^{3-\eps})$ time for $\eps>0$ when the entries of the input matrix are small integers in absolute value. Such an algorithm would be very interesting in practice, as the matrix values often represent such small discrete values.

Due to the equivalence between Min-Plus product and Maximum Subarray and since Min-Plus product can be solved in truly subcubic time when the matrix entries are small integers, it stands to reason that a truly subcubic algorithm might exist for the small entry Maximum Subarray problem as well. Unfortunately, the known reductions from Maximum Subarray to Min-Plus product blow up the matrix entries, so that even if the maximum subarray entries are in $\{-1,0,1\}$, the resulting matrices whose Min-Plus product we want to compute might have entries that are quadratic in $n$. Thus, one cannot simply use the known faster algorithms for small entry Min-Plus product to speed-up the Maximum Subarray problem with small entries.  On the lower bound end, there doesn't seem to be a way to take an instance of Min-Plus product with arbitrarily large entries and to create a maximum subarray instance with small entries. Thus, there is no obvious way to show that the small entry case is hard.

We show that Theorem~\ref{thm:low_rank} can be used to obtain a  truly subcubic algorithm for $2$D Maximum Subarray with bounded entries.

Examining
Tamaki et al. and Takaoka's reduction of Maximum Subarray to Min-Plus product, it can be seen that starting with a maximum subarray instance with entries in $\{-M,\ldots,M\}$,
one obtains $n\times n$ matrices $A$ and $B$ that are BDD as described before:
\[\forall X\in \{A,B\},\forall i,j\in [n-1],~\left|X[i,j]+X[i+1,j+1]-X[i,j+1]-X[i+1,j]\right|\leq M.\]
As BDD matrix Min-Plus product is a special case of Theorem~\ref{thm:low_rank} we immediately obtain a truly subcubic time algorithm for Maximum Subarray for matrices with entries bounded in absolute value by $O(n^{0.62})$.

\paragraph{Conditional lower bounds for $d$-Dimensional Maximum Subarray.}
Backurs et al.~\cite{backurs2016tight} showed that the $d$-Dimensional Maximum Subarray problem requires $n^{3d/2-o(1)}$ time (in the word-RAM model of computation) under the following popular hardness assumption (see e.g. \cite{vsurvey}):

\begin{hypothesis}[Max-Weight $k$-Clique Hypothesis]
In the word-RAM model with $O(\log n)$ bit words, there is no $O(n^{k-\eps})$ time algorithm for $\eps>0$ that can find a $k$-Clique of maximum weight in a given $n$-node graph with edge weights in $\{-n^{ck},\ldots,n^{ck}\}$ for large enough constant $c$.\end{hypothesis}  

The fastest known algorithm for $d$-Dimensional Maximum Subarray runs in $n^{2d-1-o(1)}$ time which is much higher than the Backurs et al.~\cite{backurs2016tight} conditional lower bound.
A natural question is thus, is there a faster algorithm for $d>2$, or can the conditional lower bounds be improved?

Our first hardness result is an improvement of the lower bound of Backurs et al., showing that Kadane's algorithm for $d$-Dimensional Maximum Subarray is conditionally tight:

\begin{theorem}\label{thm:subarray_lowerbound}
Under the Max-Weight $k$-Clique Hypothesis, in the word-RAM model with $O(\log n)$ bit words, the $d$-Dimensional Maximum Subarray problem requires $n^{2d-1-o(1)}$ time.
\end{theorem}

We were able to show that the $2$D Maximum Subarray problem can be solved faster when the matrix entries are bounded. One might wonder whether such an improvement is possible for $d>2$ as well? The simple reduction from $d$-Dimensional Maximum Subarray to $2$-Dimensional Maximum Subarray, unfortunately blows up the entries, and one cannot use the subcubic algorithm that we developed in a straightforward way. While an improvement is still possible for larger $d$, we show under a popular hardness assumption that at best one would be able to save a factor of $n^{1+o(1)}$ over the runtime of Kadane's algorithm.

The hardness assumption we use is the $\ell$-Uniform Hyperclique assumption used in prior works (see e.g. \cite{LincolnWW18,AbboudBDN18}):

\begin{hypothesis}[$\ell$-Uniform $k$-Hyperclique Hypothesis]
Let $k>\ell\geq 3$ be integers.
In the word-RAM model with $O(\log n)$ bit words, there is no $O(n^{k-\eps})$ time algorithm for $\eps>0$ that can find a hyperclique on $k$ nodes in a given $n$-node $\ell$-uniform hypergraph.\end{hypothesis}  

The hypothesis is very believable for a variety of reasons. It is known (see \cite{LincolnWW18}) that the natural extension of the techniques used to solve $k$-clique (in graphs) will not solve $k$-hyperclique in $\ell$-uniform hypergraphs faster than $n^k$. Moreover, there are known
reductions from notoriously difficult problems such as Exact
Weight $k$-Clique (a problem harder than Max Weight $k$-Clique) \cite{AbboudBDN18}, Max $\ell$-SAT and even harder Constrained
Satisfaction  Problems  (CSPs) \cite{Williams05,LincolnWW18}  to $k$-hyperclique in $\ell$-uniform hypergraphs so  that
if  the  hypothesis  is  false,  then  all  of  these  problems  have surprisingly improved algorithms. 

We prove:
\begin{theorem}\label{thm:bounded_subarray_lowerbound}
Fix any $d\geq 3$.
Under the $3$-Uniform $(2d-2)$-Hyperclique Hypothesis, in the word-RAM model with $O(\log n)$ bit words, the $d$-Dimensional Maximum Subarray problem on matrices with entries in $\{-2^{O(d)},\ldots,2^{O(d)}\}$ requires $n^{2d-2-o(1)}$ time.
\end{theorem}
That is, for any constant $d$, solving the problem in matrices with entries bounded by a constant is $n^{2d-2-o(1)}$-hard.






%% file: preliminary.tex

We use $\tilde{O}(f(n))$ to denote $f(n)\polylog n$.
For a matrix $X$, we denote by $X(i)$ the $i$th column of $X$.

The Min-Plus or $(\min,+)$-product of two $n\times n$ matrices $A$ and $B$ is the $n\times n$ matrix $C=A\star B$ with $C[i,j]=\min_k \{A[i,k]+B[k,j]\}$. The All-Pairs Shortest Paths problem (APSP) is given a graph $G=(V,E)$ with integer edge weights $w(\cdot)$, determine for all $u,v\in V$, the shortest path distance $d(u,v)$ from $u$ to $v$. 

We let $\omega$ be the exponent of square matrix multiplication, i.e. the smallest real number such that $n\times n$ matrices can be multiplied in $n^{\omega+o(1)}$ time. It is known that $2\leq \omega<2.373$ \cite{legallmult,vmmult}.

It is known \cite{AlonGM97} that for any $M \ge 1$, the Min-Plus product of two $n\times n$ matrices with entries in $\{-M,\ldots,M\}\cup\{\infty\}$ can be computed in time $\tilde{O}(Mn^\omega)$.

Our algorithm will use the following efficient data structure for half-space query in $\R^d$ for constant $d$. 
\begin{theorem}[\cite{matousek1992reporting}]\label{thm:geo}
For any constant $d \ge 2$, there exists a data structure that supports
\begin{itemize}
\item Given a set $P$ of $n$ points in $\mathbb{R}^d$, preprocess them in $\tilde{O}(n)$ time;
\item Given a halfspace $\lambda = \{x \in \mathbb{R}^d | v^T x \le  b \}$, test whether $|P \cap \lambda| > 0$ in $\tilde{O}(n^{1-1/\lfloor d/2 \rfloor}) $ time.
\item Given a halfspace $\lambda = \{x \in \mathbb{R}^d | v^T x \le  b \}$,  report all points in $P \cap \lambda$ in $\tilde{O}(n^{1-1/\lfloor d/2 \rfloor} + k) $ time, where $k = |P \cap \lambda|$. 
\end{itemize}
\end{theorem}

\section{Improvement over Min-Plus Product with One Bounded-Entry Matrix}
We slightly improve on the dependence on the entry size for computing the Min-Plus product of an arbitrary matrix and one matrix with small entries (absolute value smaller than some $M \ge 1$). Previously, the best algorithm for this runs in $\tilde{O}(Mn^{(3+\omega)/2})$ time.

\begin{proof}[Proof of Theorem \ref{thm:small_weight_times_large_weight}]
Let $\hat{C}$ be an $n \times n$ matrix, the output of our algorithm. Initialize all entries in $\hat{C}$ to $\infty$. 
Let $\Delta$ to be a small polynomial in $n$ that will be determined later. We sort each column $j$ of $B$, and arrange the elements in each column into buckets of size $\Delta$, based on the order of the elements. Specifically, the smallest $\Delta$ elements in column $j$ will be in the first bucket in column $j$, and the second smallest $\Delta$ elements will be in the second bucket, etc. We use $P_{j,\ell}$ to denote the set of row indices $k$ such that $B_{k,j}$ is in the $\ell$-th bucket of column $j$. Let the smallest entry in the $\ell$-th bucket be $S_{j,\ell}$ and let the largest entry in the $\ell$-th bucket be $L_{j,\ell}$. 

Next, for every bucket index $\ell \in [n/\Delta]$, create a matrix $B^\ell$. For the $j$-th column, if $L_{j,\ell} - S_{j,\ell} > 2M$ (large bucket), we set $B^\ell_{k, j}$ to $\infty$ for every $k$; otherwise $L_{j,\ell} - S_{j,\ell} \le 2M$ (small bucket), and we set $B^\ell_{k, j} := B_{k, j} - S_{j,\ell} - M$ for every $k \in P_{j,\ell}$, and set $B^\ell_{k, j}$ to $\infty$ for every $k \not \in P_{j,\ell}$. Notice that when $L_{j,\ell} - S_{j,\ell} \le 2M$, $B^\ell_{k, j} = B_{k, j} - S_{j,\ell} - M \in [-M, M]$ for any $k \in P_{j,\ell}$. 
Thus, we can compute $C^\ell = A \star B^\ell$ in $\tilde{O}(Mn^\omega)$ time since entries of both $A$ and $B^\ell$ are in $\{-M, \ldots, M\} \cup \{\infty\}$. We use $C^\ell_{i, j} + S_{j,\ell} + M$ to update $\hat{C}_{i,j}$. Since for every $k \in P_{j,\ell}$ when $P_{j,\ell}$ is a small bucket, $A_{i,k} + B^\ell_{k,j} + S_{j,\ell} + M = A_{i,k} + B_{k,j}$, we are essentially using $A_{i,k} + B_{k,j}$ to update $\hat{C}_{i,j}$ for every $k \in P_{j,\ell}$, if $P_{j,\ell}$ is a small bucket. Thus, after this part of the algorithm, $\hat{C}_{i,j} = \min_{k \in SB(j)} \{A_{i,k}+B_{k,j} \}$, where $SB(j)$ is the union of indices in small buckets in column $j$. This step takes $\tilde{O}(Mn^{\omega+1}/\Delta)$ time since we compute $O(n/\Delta)$ instances of Min-Plus product of two matrices whose entries are in $\{-M, \ldots, M\} \cup \{\infty\}$.

Thus, for each pair $(i, j)$, we only need to calculate $\min_{k \not \in SB(j)} \{ A_{i,k}+B_{k, j}\}$. In order to compute this, we first need to find the set of large buckets that contain an index $k$ where $A_{i,k} < \infty$. Formally, for each $i, j$, we want to find
$$\{\ell : P_{j,\ell} \text{ is a ``large'' bucket} \text{, and there exists } k \in P_{j,\ell} \text{ such that } A_{i,k} < \infty\}. $$
We can do this in $n/\Delta$ iterations. In each iteration $\ell$, we create a $\{0, \infty\}$-matrix $\bar{A}$ such that $\bar{A}_{i,k} = 0$ if and only if $A_{i,k} < \infty$. We also create a $\{0, \infty\}$-matrix $\bar{B}^\ell$ such that $\bar{B}^\ell_{k,j} = 0$ if and only if $B_{k,j}$ belongs to the $\ell$-th bucket in column $j$. The result $\bar{C}^\ell = \bar{A} \star \bar{B}^\ell$ can be computed in $O(n^\omega)$ time. If $\bar{C}^\ell_{i,j} = 0$, then bucket $P_{j,\ell}$ contains an index $k$ such that $A_{i,k} < \infty$. This step takes $\tilde{O}(n^{\omega+1}/\Delta)$ time since we compute $O(n/\Delta)$ instances of Min-Plus product with entries in $\{0, \infty\}$. 

Naively, for each pair $(i, j)$, we want to enumerate indices in all large buckets $P_{j,\ell}$ that contains an index $k$ where $A_{i,k} < \infty$. However, it is not necessary. Consider three large buckets $\ell_1 < \ell_2 < \ell_3$ (the order here means the entries in bucket $\ell_1$ are smallest, and the entries in bucket $\ell_3$ are largest). Pick any $k_1 \in P_{j,\ell_1}, k_3 \in P_{j,\ell_3}$ such that $A_{i,k_1} < \infty$ and $A_{i,k_3} < \infty$. Note that $A_{i,k_1} + B_{k_1,j} \le M + L_{j,\ell_1}$. Since buckets are ordered, the largest entry in bucket $P_{j,\ell_1}$ is at most the smallest entry in bucket $P_{j,\ell_2}$. Thus, $A_{i,k_1} + B_{k_1j} \le M + S_{j,\ell_2}$. Similarly, $A_{i,k_3} + B_{k_3,j} \ge -M + S_{j,\ell_3} \ge -M+L_{j,\ell_2}$. Since $P_{j,\ell_2}$ is a large bucket, $L_{j,\ell_2}-S_{j,\ell_2} > 2M$, which leads to $A_{i,k_1} + B_{k_1,j} < A_{i,k_3} + B_{k_3,j}$. It means that if we have two buckets $P_{j,\ell_1}$ and $P_{j,\ell_2}$ that each contains an index $k$ where $A_{i,k} < \infty$, all buckets that are larger than them won't give a better candidate $k$.
Therefore, for each $(i, j)$, we only need to enumerate the first two large buckets that contain indices $k$ where $A_{i,k} < \infty$. Thus, it takes $\tilde{O}(n^2\Delta)$ time to cover large buckets. 

In total, the running time of the algorithm is $\tilde{O}(Mn^{1+\omega}/\Delta + n^2 \Delta)$. Setting $\Delta=\sqrt{M} n^{(\omega - 1) / 2}$ gives the claimed $\tilde{O}(\sqrt{M} n^{(3+\omega)/2})$ time. 
\end{proof}






%% file: upper_bound.tex

Let $\Delta$ be a positive integer that is a small polynomial in $n$. Assume for simplicity that $n$ is a multiple of $\Delta$. Then we can partition $[n]$ into $n/\Delta$ groups by setting $I(i') = \{i: i' - \Delta < i \le i'\}$ for any $i'$ divisible by $\Delta$. For any $i', j'$ that are multiples of $\Delta$, we can group all entries $A_{i,j}$ where $i \in I(i'), j \in I(j')$ into a sub-matrix of size $\Delta \times \Delta$, thus partitioning $A$ into sub-matrices of size $\Delta\times\Delta$. We can similarly partition $B$ into sub-matrices of size $\Delta \times \Delta$. 

In Theorem~\ref{thm:low_rank_messy} below we will show that if each of the $\Delta\times\Delta$  sub-matrices of $B$ are close in $\ell_\infty$ norm to an $O(1)$-rank matrix, then we can compute $A \star B$ in truly sub-cubic time. In other words, we need the blocks of $B$ to have constant $n^\eps$-approximate rank for small $\eps>0$.

\begin{theorem}\label{thm:low_rank_messy} Let  $A, B$ be
two given $n\times n$ matrices whose entries are $\polylog n$ bit integers.
Let $W$ be a nonnegative integer and let $d\geq 1$ be an integer with $d=O(1)$.
Suppose that for all $k', j'$ multiples of $\Delta$, we can find two $d$ by $\Delta$ integer matrices $X_{k',j'}$ and $Y_{k',j'}$, such that for any $(k, j) \in I(k') \times I(j')$, $\left| B_{k, j} - X_{k',j'}(k)^T Y_{k',j'}(j) \right| \le W$. Then, for any integer $\rho\geq 1$, there exists a $$\tilde{O}(n^3  \cdot \Delta^{-1/\lfloor (d+1)/2 \rfloor} + \rho \sqrt{W} n^{(3+\omega)/2} + n^3/\rho)$$ time algorithm that computes $A \star B$. 
\end{theorem}

To obtain Theorem~\ref{thm:low_rank} from Theorem~\ref{thm:low_rank_messy}, we set $\rho$ to $\lceil n^{(3-\omega)/4} W^{-1/4}\rceil$ when $W\leq n^{3-\omega}$; otherwise we can run the trivial cubic time algorithm for Min-Plus product.

The algorithm starts with the framework behind the Bringmann et al. algorithm \cite{BringmannGSW16} that computes the $(\min, +)$-product of two matrices with bounded differences.
However, each of the three steps in the framework requires a completely different approach due to the less structured nature of matrix $B$. 
The resulting algorithm is a strong generalization of the algorithm of \cite{BringmannGSW16}. 

In the rest of this section, we will use $C = A\star B$ to denote the desired $(\min, +)$-product, and use $\hat{C}$ as the output of our algorithm. The algorithm contains three phases. In the first phase, we will compute a matrix $\tilde{C}$, such that every entry of $\tilde{C}$ is an additive approximation of the corresponding entry in the desired output $C$. In the second phase, we will compute $\hat{C}$ by calculating the $(\min, +)$-product of some small weight matrices generated by $A, B$ and $\tilde{C}$ using fast matrix multiplication. In the third phase, we will correct all entries of $\hat{C}$ by efficiently enumerating all $A_{ik}+B_{kj}$ that can possibly improve $\hat{C}_{ij}$. 

\subsection{Phase 1: Approximated Min-Plus Product}

For each triple $(i', k', j')$ such that all $i', k', j'$ are multiples of $\Delta$, if we can compute an additive approximation $\tilde{C}^{i',k',j'}$ of the $(\min, +)$-product $A_{I(i'),I(k')} \star B_{I(k'),I(j')}$, then we can, in $O(n^3/\Delta)$ time, compute $\tilde{C}_{i,j}=\min_{k': \Delta | k'} \tilde{C}^{i',k',j'}_{i,j}$ where $i \in I(i'), j\in I(j')$. We will use the geometric data structure from Theorem \ref{thm:geo} to approximate $A_{I(i'),I(k')} \star B_{I(k'),I(j')}$. 

\begin{lemma}\label{lem:approx_Delta_MM}
There exists a $\tilde{O}(\Delta^{3-1/\lfloor (d+1)/2\rfloor})$ time algorithm that computes a $W$-additive approximation $\tilde{C}^{i',k',j'}$ of $A_{I(i'),I(k')} \star B_{I(k'),I(j')}$, for any $i', k', j'$ multiples of $\Delta$. 
\end{lemma}
\begin{proof}
By the structure of $B$, for any $(k, j) \in I(k') \times I(j')$, we have 
$$\left|B_{k, j} - X_{k',j'}(k)^T Y_{k',j'}(j)\right| \le W.$$
Therefore, if we can accurately compute 
$$\tilde{C}^{i',k',j'}_{i, j} = \min_{k \in I(k')} \left\{A_{i, k} + X_{k',j'}(k)^T Y_{k',j'}(j) \right\}, $$
we immediately get a $W$-additive approximation of $A_{I(i'),I(k')} \star B_{I(k'),I(j')}$. 

Create a set of $(d+1)$-dimensional points
$$P_i = \left\{\begin{pmatrix}
A_{i, k}\\
X_{k',j'}(k)
\end{pmatrix}: k \in I(k') \right\},$$
and use the data structure in Theorem \ref{thm:geo} to pre-process this set. Each set has size $O(\Delta)$, and there are $|I(i')| = \Delta$ such sets, so the total pre-processing time is $\tilde{O}(\Delta^2)$. For any $j \in I(j')$, we create a $(d+1)$-dimensional vector $v_j=\begin{pmatrix}
1\\
Y_{k',j'}(j)
\end{pmatrix}$. We observe that 

$$A_{i,k} + X_{k', j'}(k)^TY_{k',j'}(j) = v_j^T \begin{pmatrix}
A_{i, k}\\
X_{k',j'}(k)
\end{pmatrix}, $$
so $\tilde{C}^{i',k',j'}_{i, j} = \min_{x \in P_i} v_j^T x$. In order to compute $\min_{x \in P_i} v_j^T x$ for every pair $(i, j)$, we use the emptiness query of the geometric data structure. We want to find the minimum value of $b$, so that there exists a point $x \in P_i$ where $v_j^T x \le b$. This is equivalent to testing whether the half-space $\lambda = \{x\in \mathbb{R}^{d+1} | v_j^T x \le b\}$ intersects $P_i$. Therefore, we can use binary search on the minimum value of $b$, which will be equal to $\tilde{C}^{i',k',j'}_{i, j}$. 

Each emptiness query takes $\tilde{O}(\Delta^{1-1/\lfloor (d+1)/2 \rfloor}) $ time, and we need to query $O(\log (|A|_\infty + |B|_\infty))$ time for each pair $(i, j) \in I(i') \times I(j')$, so in total it takes $\tilde{O}(\Delta^{3-1/\lfloor (d+1)/2 \rfloor}) $ time to compute $\tilde{C}^{i',k',j'}$. 
\end{proof}

\begin{lemma}\label{lem:approx_MM}
There exists a $\tilde{O}(n^3 \cdot \Delta^{-1/\lfloor (d+1)/2\rfloor})$ time algorithm that computes a $W$-additive approximation $\tilde{C}$ of $A \star B$. 
\end{lemma}
\begin{proof}
For every triple $(i', k', j')$ where $i', k', j'$ are multiples of $\Delta$, we compute $\tilde{C}^{i',k',j'}$ using the algorithm in Lemma \ref{lem:approx_Delta_MM}. Since there are $O((n/\Delta)^3)$ such triples, it takes $\tilde{O}(n^3 \cdot \Delta^{-1/\lfloor (d+1)/2\rfloor})$ time in total. Then we compute $\tilde{C}$ using 
$\tilde{C}_{i,j}=\min_{k': \Delta | k'} \tilde{C}^{i',k',j'}_{i,j}$ in $O(n^3 \cdot \Delta^{-1})$ time.
\end{proof}

\subsection{Phase 2: Create Estimate Matrix $\hat{C}$ by Random Sampling}

This phase of the algorithm consists of $10 \rho \ln n$ rounds. For each round $r$, we sample $j^r \in [n]$ uniformly at random. Define $A^r$ to be an $n \times n$ matrix where $A^r_{i,k}:=A_{i,k}+B_{k, j^r} - \tilde{C}_{i,j^r}$, and define $B^r$ such that $B^r_{k,j}:=B_{k, j} - B_{k,j^r}$. If we compute $C^r = A^r \star B^r$, we can infer $C = A \star B$ via the relation $C_{i,j}=C^r_{i,j}+\tilde{C}_{i,j^r}$. However, it is not always possible to compute $C^r$ efficiently, since the weights of $A^r$ and $B^r$ can be arbitrarily large. Therefore, we need to set the large entries in $A^r$ to be $\infty$ in order to compute $A^r \star B^r$ efficiently. Specifically, we will set an entry of $A^r$ to $\infty$ if its absolute value is more than $3W$. Then we can compute $C^r = A^r \star B^r$ in $\tilde{O}(\sqrt{W} n^{(3+\omega)/2})$ time by Theorem \ref{thm:small_weight_times_large_weight}. 

This phase deviates from the approach of Bringmann et al. Bringmann et al. set the large entries of both $A^r$ and $B^r$ to $\infty$. If we were to do that, we wouldn't be able to complete Phase 3 -- there doesn't seem to be enough to finish the $(\min, +)$-product computation in truly subcubic time. By only setting the large entries of $A^r$ to $\infty$ and letting $B^r$ keep all its entries, we offload enough work onto Phase 2, so that now Phase 3 can also be done in truly subcubic time.

Since there are $\rho$ rounds, the total time complexity of this phase is $\tilde{O}(\rho\sqrt{W} n^{(3+\omega)/2})$. Intuitively, fix any $i, j \in [n]$, if $A^r_{i, k}$ is not set to $\infty$, then $C^r_{i, j} \le \left(A_{i,k}+B_{k, j^r} - \tilde{C}_{i,j^r}\right) + \left(B_{k, j} - B_{k,j^r}\right) = A_{i, k} + B_{k, j} - \tilde{C}_{i,j^r}$. Thus, if we take $\hat{C}_{i, j}$ to be $\min_r \left\{ C^r_{i, j}+ \tilde{C}_{i,j^r} \right\}$, then $\hat{C}_{i, j} \le A_{i, k} + B_{k, j}$ as long as $A^r_{i, k} < \infty$ for at least one $r$. We will formalize this intuition and show that we only need to enumerate a sub-cubic number of $(i, k, j)$ triples in order to correct all entries in $\hat{C}$ after $10 \rho \ln n$ rounds. 
\begin{definition}
We call a triple $(i, k, j)$
\begin{itemize}
\item strongly relevant if $A_{i,k}+B_{k, j} = C_{i, j }$;
\item weakly relevant if $|A_{i,k}+B_{k, j} - \tilde{C}_{i, j }| \le 3W$;
\item uncovered if for all $1 \le r \le 10 \rho \ln n$, $|A_{i,k}^r|>3W$.
\end{itemize}
\end{definition}
Since whether a triple $(i, k, j)$ is uncovered only depends on $(i, k)$, we will also call a pair $(i, k)$  uncovered if for all $1 \le r \le 10 \rho \ln n$, $|A_{i,k}^r|>3W$. A triple (pair) that is not uncovered will be called \textit{covered}. 

If a triple $(i, k, j)$ is not strongly relevant, then even if $A^r_{i, k} = \infty$ for every round $r$, it doesn't affect whether $\hat{C}_{i,j} = C_{i,j}$. If a triple $(i, k, j)$ is covered, then there exists a round $r$ such that $A^r_{i, k}$ is not set to $\infty$. In this case, $\hat{C}_{i, j} \le C^r_{i,j}+\tilde{C}_{i,j^r} \le A_{i,k}+B_{k, j}$. Since only strongly relevant triples matter, and our algorithm already updates the answer for every covered triples, so we need to update $\hat{C}$ using triples that are both strongly relevant and uncovered. Specifically, if we can enumerate all strongly relevant and uncovered triples $(i, k, j)$, and update $\hat{C}_{i,j}$ using $A_{i,k}+B_{k,j}$, we can correct all entries in $\hat{C}$. 

However, it is hard to only enumerate strongly relevant and uncovered triples without enumerating some additional triples. Thus we allow the algorithm to enumerate some of the \textit{weakly} relevant and uncovered triples, in addition to strongly relevant and uncovered triples. In this way, we can cover all strongly relevant and uncovered triples, while keeping the total number of triples small. Note that since $\tilde{C}$ is a $W$-additive approximation of $C$, a strongly relevant triple is always weakly relevant, so we care about the total number of weakly relevant and uncovered triples. 
The next lemma shows that the number of such  triples is truly sub-cubic. 
 
 \begin{lemma}\label{lem:weak_triple_small}
With high probability, the number of \textit{weakly} relevant and uncovered triples is at most $n^3/\rho$.
 \end{lemma}
 \begin{proof}
We say a pair $(i, k)$ is bad if the number of weakly relevant triples $(i, k, j)$ is greater than $n/\rho$. 

 Fix any bad $(i, k)$. For a random $j \in [n]$, the probability that $(i, k, j)$ is weakly relevant is at least $1/\rho$. Since we have $10 \rho \ln n$ randomly sampled $j^r$, the probability that at least one $j^r$ forms a weakly relevant triple $(i, k, j^r)$ is at least $1- \left(1 - 1/\rho\right)^{10 \rho \ln n} \ge 1-1/n^{10}$. Suppose $(i, k, j^r)$ is weakly relevant, then $|A^r_{i, k}|=|A_{i,k}+B_{k, j^r} - \tilde{C}_{i, j^r}| \le 3W$. 
Thus, $A^r_{i, k}$ will not be set to $\infty$ in round $r$, so $(i, k)$ is covered. 
By taking a union bound over all bad $(i, k)$, we conclude that with probability at least $1-1/n^8$, all triples $(i, k, j)$ will be covered when $(i, k)$ is bad. It means that with high probability, these bad $(i, k)$ pairs don't contribute any weakly relevant and uncovered triples. 

For a pair $(i, k)$ that is not bad, the number of $j$ such that $(i, k, j)$ is weakly relevant is at most $n/\rho$, by definition of a bad pair. Thus, these $(i, k)$ pairs contribute at most $n^3/\rho$ weakly relevant and uncovered pairs. 
 \end{proof}
 
\subsection{Phase 3: Enumerate Strongly Relevant and Uncovered Triples}
 
It remains to show how to quickly iterate through strongly relevant, uncovered triples. Fix $i', k', j'$ multiples of $\Delta$, we will show how to efficiently enumerate strongly relevant, uncovered triples in $I(i') \times I(k') \times I(j')$. We consider the set $S_{i',k',j'} \subseteq I(i') \times I(j') \times I(k')$, consisting of triples $(i, j, k)$ such that $A_{i,k} + X_{k',j'}(k)^T Y_{k',j'}(j) \le 2W+\tilde{C}_{i,j}$. The following lemma shows that it is sufficient to enumerate triples in this set. 

\begin{lemma}\label{lem:Sikj}
The set $S_{i',k',j'}$ contains all strongly relevant triples in $I(i') \times I(j') \times I(k')$, and it contains only weakly relevant triples. 
\end{lemma}
\begin{proof}
Let $(i, k, j)$ be any strongly relevant triple. Then 
\begin{equation*}
    \begin{split}
        &A_{i,k}+X_{k',j'}(k)^T Y_{k',j'}(j)-\tilde{C}_{i,j}\\
        =& A_{i,k} + B_{k,j} - C_{i,j} + \left(X_{k',j'}(k)^T Y_{k',j'}(j)-B_{k,j}\right) + \left(C_{i,j}-\tilde{C}_{i,j}\right)\\
         \le & 2W,
    \end{split}
\end{equation*}
so $(i, k, j) \in S_{i',k',j'}$. 

In order to prove the second claim, we need to show $\left|A_{i,k}+B_{k,j} - \tilde{C}_{i,j}\right| \le 3W$ for every triple $(i, j, k) \in S_{i',k',j'}$. Since $\tilde{C}$ is a $W$-additive approximation of $C$,  $A_{i,k}+B_{k,j} - \tilde{C}_{i,j} \ge -W$ holds for every triple $(i, k, j)$. 
Since $(i,k,j)\in S_{i',k',j'}$, we have $A_{i,k} + X_{k',j'}(k)^T Y_{k',j'}(j) \le 2W+\tilde{C}_{i,j}$, or equivalently:
$$A_{i,k}+B_{k,j} - \tilde{C}_{i,j}  \le 2W+(B_{k,j}-X_{k',j'}(k)^T Y_{k',j'}(j)).$$
Since $B_{k,j}$ differs at most $W$ from $X_{k',j'}(k)^T Y_{k',j'}(j)$, we have $A_{i,k}+B_{k,j}-\tilde{C}_{i,j} \le 3W$. 
\end{proof}

By Lemma \ref{lem:Sikj}, it suffices to enumerate uncovered triples in $S_{i',k',j'}$. 
For each $i \in I(i')$, create a set of $(d+1)$-dimensional points $$Q_i = \left\{\begin{pmatrix}
A_{i, k}\\
X_{k',j'}(k)
\end{pmatrix}: k \in I(k') \wedge (i, k) \text{ is uncovered}\right\}, $$
and pre-process these points using the data structure in Theorem \ref{thm:geo}. For each $(i, j) \in I(i') \times I(j')$, we create the following half-space:
$$\lambda_{i,j}=\left\{x\in \mathbb{R}^{d+1}| \begin{pmatrix}
1\\
Y_{k',j'}(j)
\end{pmatrix}^T x \le 2W + \tilde{C}_{i,j}\right\}.$$
Then $Q_i \cap \lambda_{i,j}$ contains the set of $k \in I(k')$ such that $(i, k, j) \in S_{i',j',k'}$ and $(i, k)$ is uncovered. Therefore, we can use the data structure in Theorem \ref{thm:geo} to list the set of $k$ in $\tilde{O}(\Delta^{1-1/\lfloor (d+1)/2\rfloor} + |Q_i \cap \lambda|)$ time. Note that the total number of listed points is bounded by the number of weakly-relevant, uncovered triples, so the summation of the second term over all $i', k', j', i, j$ is $\tilde{O}(n^3/\rho)$. The summation of the first term over all $i', k', j', i, j$ is $\tilde{O}(n^3  \cdot \Delta^{-1/\lfloor (d+1)/2 \rfloor})$.

%% file: apsp.tex
In this section, we study an algorithm for APSP where the edge weights of the input graph can be approximated by a low dimensional geometric function. 

Let $W$ be an integer, $d\geq 1$ be a constant integer and let $\delta\in (0,1]$ be a constant. Let us define (as in the introduction) a $(W,d,\delta)$-geometrically weighted clustered graph, $(W,d,\delta)$-GWC for short as follows. $G=(V, E)$ is $(W,d,\delta)$-GWC if
\begin{itemize}
\item $V$ is partitioned into $t = n^{1-\delta}$ subsets $V_1, V_2, \ldots, V_t$ of size $O(n^\delta)$,
\item for every $i,j\in \{1,\ldots,t\}$, each $v \in V_i$ is assigned a $d$-dimensional integer vector $p^{i,j}(v)$, and each $u \in V_j$ is assigned a $d$-dimensional integer vector $q^{i,j}(u)$, and
\item for $v \in V_i$,  $u \in V_j$, $|w(v,u)-p^{i,j}(v)^T q^{i,j}(u)|\leq W.$ In other words, the edge weights in $V_i\times V_j$ are determined by a matrix whose $W$-approximate rank is at most $d$,
\item the absolute value of any edge weight is at most $O(n^c)$ for some constant $c$.
\end{itemize}

The last bullet is only needed so that SSSP in such graphs can be performed in truly subcubic time even if there are negative edge weights, e.g. as in Goldberg~\cite{Goldberg95}.

The following is a direct corollary of Theorem \ref{thm:low_rank}:
\begin{corollary}\label{cor:geo_mul}
For any integer matrix $A$ and $B$ the generalized adjacency matrix of a $(W,d,\delta)$-GWC graph, we can compute $C= A \star B$ in $\tilde{O}(n^{3-\delta/ \lfloor(d+1)/2 \rfloor} + n^{(9+\omega)/4} \cdot W^{1/4})$ time. 
\end{corollary}

Using Corollary \ref{cor:geo_mul}, we can compute the shortest distance between two vertices among all paths with a small length. Using a standard technique in APSP algorithms, we can compute shortest paths among the long paths by randomly sampling vertices. 

\begin{theorem}
We can compute APSP for a $(W,d,\delta)$-GWC graph in 
\begin{itemize}
\item $\tilde{O}(W^{1/8} n^{(21+\omega)/8})$ time whenever $W>n^{3-\omega-4\delta/\lfloor (d+1)/2 \rfloor}$, and
\item $\tilde{O}(n^{3-\delta/(2\lfloor (d+1)/2 \rfloor}))$ time if $W\leq n^{3-\omega-4\delta/\lfloor (d+1)/2 \rfloor}$. 
\end{itemize}
\end{theorem}
\begin{proof}
Let $B$ be the generalized adjacency matrix, and let $\ell$ be a parameter to be fixed later. For each $i \le \ell$, we can compute $B^{(i)}$ by iterating the product $B^{(i)}\leftarrow A\star B$ for $A=B^{(i-1)}$. By Corollary \ref{cor:geo_mul}, this step will take $\tilde{O}(\ell \cdot n^{3-\delta/ \lfloor(d+1)/2 \rfloor} + \ell \cdot n^{(9+\omega)/4} \cdot W^{1/4})$ time. 

We can randomly sample $\tilde{\Theta}(n/\ell)$ vertices $S$, and perform Dijkstra's algorithm to and from these vertices in $S$ (after the usual Johnson's preprocessing to get rid of any negative weights, and using say Goldberg's SSSP algorithm which works in truly subcubic time since the edge weights are assumed to be polynomial in $n$). With high probability, $S$ hits a shortest path between every two vertices that have a shortest path containing at least $\ell$ vertices. We can perform this step in $\tilde{\Theta}(n^3/\ell)$ time.

The first step gives the shortest path between two vertices that uses at most $\ell$ vertices, and the second step gives the shortest path that uses more than $\ell$ vertices. Thus, by taking the smaller one over these two, we can correctly compute the APSP.

The total time complexity is $\tilde{O}(\ell \cdot n^{3-\delta/ \lfloor(d+1)/2 \rfloor} + \ell \cdot n^{(9+\omega)/4} \cdot W^{1/4} + n^3/\ell)$.

If $W>n^{3-\omega-4\delta/\lfloor (d+1)/2 \rfloor}$, then $n^{(9+\omega)/4} \cdot W^{1/4}>n^{3-\delta/ \lfloor(d+1)/2 \rfloor}$, so the running time is 
\[\tilde{O}(\ell \cdot n^{(9+\omega)/4} \cdot W^{1/4} + n^3/\ell).\]
We can set $\ell$ to be $n^{(3-\omega)/8}/W^{1/8}$, balancing the two terms of the runtime and thus minimizing it at $\tilde{O}(W^{1/8} n^{(21+\omega)/8})$.

Otherwise, if $W\leq n^{3-\omega-4\delta/\lfloor (d+1)/2 \rfloor}$, then $n^{(9+\omega)/4} \cdot W^{1/4}\leq n^{3-\delta/ \lfloor(d+1)/2 \rfloor}$, so the running time is 
\[\tilde{O}(\ell \cdot n^{3-\delta/ \lfloor(d+1)/2\rfloor}+ n^3/\ell).\]
Then it makes sense to set $\ell=n^{\delta/ (2\lfloor(d+1)/2\rfloor)}$, minimizing the runtime to $\tilde{O}(n^{3-\delta/ (2\lfloor(d+1)/2\rfloor}))$.

\end{proof}


%% file: range_mode.tex

In this section, as an application of our Main Algorithm, we give an $O(n^{1.5-\eps})$ time  algorithm for the Batch Range Mode query problem for some $\eps>0$. In a high level, there are two steps in the algorithm. First we use the Main Algorithm to obtain a truly subcubic time $(\min, +)$-product for particularly structured matrices; then we show how to reduce range mode to this kind of structured $(\min, +)$-product.

\begin{lemma}\label{lem:lem_multi_inc}
Let $A, B$ be two $n \times n$ integer matrices, where matrix $B$ meets
\begin{enumerate}[label=\arabic*)]
    \item Each row of $B$ is non-increasing; 
    \item The difference between the sum of elements in the $j$-th column, and the sum of elements in the $(j+1)$-th column is at most $m$, for any $j$.
\end{enumerate}
When $m = \Omega(n^{(\omega-1)/2})$, there exists a $\tilde{O}(n^{(14+\omega)/6} \cdot m^{1/6})$ time algorithm that computes $A \star B$, which is truly sub-cubic as long as $m=O(n^{4-\omega-\eps})$ for some $\eps>0$. When $m=O(n^{(\omega-1)/2})$, there exists a $\tilde{O}(n^{(9+\omega)/4})$ time algorithm. 
\end{lemma}
\begin{proof}
Let $\Delta, \gamma \ge 1$ be small polynomials in $n$ to be fixed later. Fix $j'$ a multiple of $\Delta$. Since $\sum_{k=1}^n B_{k,j} - \sum_{k=1}^n B_{k,j+1} \le m$ for any $j \in I(j')$, we have
$$\sum_{k=1}^n B_{k,j
'-\Delta+1} - \sum_{k=1}^n B_{k,j'}\le \Delta m. $$
By averaging, there are at most $\Delta m / \gamma$ indices $k \in [n]$ such that $B_{k,j
'-\Delta+1}-B_{k,j'} \ge \gamma$. For each $k$ such that $B_{k,j'
-\Delta+1}-B_{k,j'}  \ge \gamma$, and for each $j \in I(j')$, we set $B_{k,j}$ as $M$, for some large enough integer $M$ (larger than all entries in $B$). We call the matrix $\hat{B}$ after we do this transformation for every $j'$. Note that $\hat{B}$ differs with $B$ in at most $n m \Delta / \gamma$ entries. 

Notice that $\hat{B}$ has the following nice property: for each $j', k'$ multiples of $\Delta$, $\left| \hat{B}_{j, k} - \hat{B}_{j', k} \right| \le \gamma$ for any $j \in I(j'), k \in I(k')$. Consider a set of $1$-dimensional vectors $X_{k',j'}(k) = [\hat{B}_{j', k}]$, and $Y_{k',j'}(j) = [1]$, then $\left| \hat{B}_{j, k} - X_{k',j'}(k)^T Y_{k',j'}(k) \right| \le \gamma$. Therefore, we can apply Theorem \ref{thm:low_rank} using $d=1$. This gives a
$$\tilde{O}(n^3  /\Delta + n^{(9+\omega)/4} \cdot \gamma^{1/4})$$
time algorithm to compute $\hat{C} = A \star \hat{B}$.

We can recover $C = A \star B$ from $\hat{C}$. Since $B$ and $\hat{B}$ differ in at most $nm \Delta / \gamma$ entries, and $\hat{B}$ is larger on these entries, we can enumerate $A_{i,k}+B_{k,j}$ to update $C_{i,j}$, where $B_{k,j}$ differs from $\hat{B}_{k,j}$. This will take $O(n^2 m \Delta / \gamma)$ time. 

The total complexity is $\tilde{O}(n^3  /\Delta +n^{(9+\omega)/4} \cdot \gamma^{1/4} + n^2 m \Delta / \gamma)$. 

When $m=\Omega(n^{(\omega-1)/2})$, we can balance by setting $\Delta = n^{(4-\omega)/6} m^{-1/6}$, and $\gamma = n^{(1-\omega)/3} m^{2/3}$. This gives a $\tilde{O}(n^{(14+\omega)/6} \cdot m^{1/6})$ time algorithm.

When $m=O(n^{(\omega-1)/2})$, we can balance by setting $\Delta = n^{(3-\omega)/4} $, and $\gamma = 1$ to get a $\tilde{O}(n^{(9+\omega)/4})$ time algorithm.
\end{proof}

\begin{theorem}
Given a sequence $a_1, a_2, \ldots, a_n$, and $n$ ranges $[l_1, r_1], [l_2, r_2], \ldots, [l_n, r_n]$, there exists a $\tilde{O}(n^{(27+2\omega)/(19+\omega)})$ time algorithm that computes the frequency of the most frequent element for each range $[l_i, r_i]$. Using $\omega \le 2.373$, this algorithm runs in $\tilde{O}(n^{1.4854})$ time. 
\end{theorem}
\begin{proof}

Without loss of generality, we assume $l_i \le n/2 < r_i$. Otherwise, we can use a divide-and-conquer approach to first compute the queries that satisfy $l_i \le n/2 < r_i$, then recurse on the two halves $[1,n/2]$ and $(n/2,n]$ to compute answers. Since the proposed time complexity is $\Omega(n^{1+\eps})$ for some $\eps > 0$, the total time complexity does not change by the Master Theorem. 

Let $T$ be a parameter of the algorithm that controls the block size as well as a threshold frequency for frequent elements and infrequent elements. 
We handle elements that appear at most $T$ times (infrequent elements), and elements that appear more than $T$ times (frequent elements) differently. 

Fix some infrequent elements $x$. For any $a_j = a_k = x$ where $j \le k$, we create an interval $[j, k]$, whose weight is the number of occurrence of $x$ in the range $[j, k]$. Since $x$ occurs at most $T$ times, the number of of such intervals is at most $O(Tn)$. To query the largest frequency in a range $[l_i, r_i]$, it is equivalent to ask the largest weight of intervals that are contained in the interval $[l_i, r_i]$. This problem can be solved by, for instance, using a persistent balanced search tree, in $\tilde{O}(Tn)$ preprocess time and $\tilde{O}(1)$ query time. 

Now consider the ``frequent'' elements in the array that occur more than $T$ times. There are at most $n/T$ distinct frequent elements in the array. For each of these elements $x$, we create a balanced binary search tree $B_x$, whose elements are the set of occurrences $\{i: a_i = x\}$, augmented with the size of the subtree rooted at each node. 
We split the whole sequence $a_1,\ldots,a_n$ into consecutive blocks of size $O(T)$, so that $n/2$ is the right boundary of one block and the left boundary of the next block. 

For a range $[l_i, r_i]$, let $S_s, S_{s+1}, \ldots, S_{t}$ be the maximum set of blocks in this range, then the range mode of $[l_i, r_i]$ is either the range mode of the subinterval $S_s, S_{s+1}, \ldots, S_{t}$, or some elements in $[l_i, r_i] \setminus \left\{S_s, S_{s+1}, \ldots, S_{t}\right\}$.

Suppose that the range mode of $[l_i,r_i]$ is not the range mode of $S_s, S_{s+1}, \ldots, S_{t}$.
Then, we have a candidate list of $O(T)$ numbers (those to the left and right of $S_s, S_{s+1}, \ldots, S_{t}$ in $[l_i,r_i]$)  that can possibly be the range mode of the interval $[l_i, r_i]$. For each of these numbers $x$, we can query its occurrence in the range $[l_i, r_i]$ by querying  the number of elements between $[l_i, r_i]$ in $B_x$ which takes $O(\log n)$ time due to the augmentation. 

Therefore, it takes $\tilde{O}(T)$ overhead to compute the range mode of $[l_i, r_i]$ once we know the range mode of $S_s, S_{s+1}, \ldots, S_{t}$. Thus, we can focus on the sub-problem of computing the range mode of the subinterval $S_s, S_{s+1}, \ldots, S_{t}$, where $S_s$ is to the left of $n/2$, and $S_t$ is to the right of $n/2$ and some pair of blocks $S_{i*},S_{i*+1}$ end and start (respectively) at $n/2$. Call these last two the middle blocks.

We create two matrices $A$ and $B$. The columns of $A$ and rows of $B$ are indexed by the heavy elements in $a_1,\ldots,a_n$. The rows of $A$ and columns of $B$ are indexed by $j$ such that $S_j$ is one of the blocks of size $T$ that we partitioned $a_1,\ldots, a_n$ into. Hence both $A$ and $B$ are $O(n/T)$ by $O(n/T)$ matrices.

More concretely, for each $S_s$ to the left of $n/2$, we create a row $s$ in matrix $A$, where $A_{s,k}$ is the negated number of occurrences of element $k$ in the subinterval $S_s, \ldots, S_{i*}$ (recall that $S_{i*}$ ends at $n/2$); for each $S_t$ to the right of $n/2$, we create column $t$ in matrix $B$ where $B_{k,t}$ is the negated number of occurrences of element $k$ in the subinterval  $S_{i*+1},\ldots, S_t$ (recall that $S_{i*+1}$ starts at $n/2+1$). Therefore, the negated Min-Plus product entry $-(A\star B)_{s,t}$ will be the range mode in the full subinterval $S_s, S_{s+1}, \ldots, S_{t}$. 

Note that $A, B$ are $O(n/T)$ by $O(n/T)$ matrices, each row of $B$ is monotonically non-increasing, and the difference between the $i$-th column and $(i+1)$-th column is at most $T$. Therefore, we can apply Lemma~\ref{lem:lem_multi_inc} to multiply $A \star B$ in $\tilde{O}((n/T)^{(14+\omega)/6} T^{1/6})$ time when $T = \Omega((n/T)^{(\omega-1)/2})$.

Therefore, the overall running time of the algorithm is $\tilde{O}((n/T)^{(14+\omega)/6} T^{1/6} + nT)$. By setting $T=n^{{(8+\omega)/(19+\omega)}}$, we get a $\tilde{O}(n^{(27+2\omega)/(19+\omega)})$ time algorithm.
\end{proof}


%% file: subarray.tex

In \cite{tamaki1998algorithms}, Tamaki and Tokuyama reduced 2D maximum subarray problem to $(\min, +)$-product of two matrices $A$, $B$, using a divide-and-conquer approach. In this reduction, if the absolute values of the entries of the input array are bounded by $M$, then the matrix $A$ has the property that $$\forall i, j, \left| A_{i+1,j+1} - A_{i, j+1} - A_{i+1,j} + A_{i,j} \right| \le M.$$
The same property holds for $B$ as well. If we can compute $(\min, +)$-product of matrices with this property in sub-cubic time, then we can solve the maximum subarray problem with bounded entry in sub-cubic time as well. 

Motivated by this application, we define the following notion of finite difference operator. 
\begin{definition}
The finite difference operator $\mathcal{D}$ acts on a matrix such that
$$(\mathcal{D}A)_{i,j} = A_{i+1,j+1} -A_{i, j+1} - A_{i+1,j} + A_{i, j}.$$
\end{definition}
Using this definition, we can rephrase the property of matrices related with the maximum subarray problem as $|(\mathcal{D}A)_{i,j}| \le M$. 

In the rest of this section, we will show how to compute $A \star B$ in sub-cubic time when $\left| (\mathcal{D}^t B)_{i,j} \right| \le M$ for some constant $t$. The following lemma shows that matrices with bounded entries after the operator $\mathcal{D}^t$ can be approximated with a low rank matrix. 

\begin{lemma}\label{lem:D^k}
For an arbitrary matrix $B$ where $| (\mathcal{D}^t B)_{i,j}| \le M$,  there exist $2n$ integer vectors of $(2t)$-dimension  $X(1), X(2), \ldots, X(n)$ and $Y(1), Y(2), \ldots, Y(n)$, such that $\left| B_{i,j}-X(i)^TY(j)\right| = O(M n^{2t})$.
 
\end{lemma}
\begin{proof}
We prove this by induction on $t$. When $t=0$, the claim is trivially true.

When $t>0$, assume the claim is true for $t-1$. Let $A = \mathcal{D} B$. Since $\mathcal{D}^{t-1} A = \mathcal{D}^t B$, by induction, there exists $(2t-2)$-dimension vectors $P(i), Q(j)$ such that $\left| A_{i,j} - P(i)^TQ(j)\right| = O(M n^{2t-2})$. Define $E_{i,j}= A_{i,j} - P(i)^TQ(j)$ to be the error term, whose absolute value is bounded by $O(Mn^{2t-2})$. Since $A = \mathcal{D} B$, 
\begin{equation*}
\begin{split}
B_{i,j} &= \left(\sum_{a = 1}^{i-1} \sum_{b=1}^{j-1} A_{a,b}\right) - B_{1,1}+B_{i,1}+B_{1,j}\\
&= \left(\sum_{a = 1}^{i-1} \sum_{b=1}^{j-1} \left(P(a)^TQ(b) + E_{a,b}\right)\right) - B_{1,1}+B_{i,1}+B_{1,j}\\
&= \left(\sum_{a = 1}^{i-1}P(a) \right)^T \left( \sum_{b=1}^{j-1} Q(b)\right) - B_{1,1}+B_{i,1}+B_{1,j} + \left(\sum_{a = 1}^{i-1} \sum_{b=1}^{j-1} E_{a,b}\right)\\
&= 
\begin{bmatrix}
1\\
-B_{1,1}+B_{i,1}\\
\sum_{a = 1}^{i-1}P(a)
\end{bmatrix}^T
\begin{bmatrix}
B_{1,j}\\
1\\
\sum_{b = 1}^{j-1}Q(b)
\end{bmatrix}
+ \left(\sum_{a = 1}^{i-1} \sum_{b=1}^{j-1} E_{a,b}\right)
\end{split}
\end{equation*}
Therefore, if we set $$X(i):= \begin{bmatrix}
1\\
-B_{1,1}+B_{i,1}\\
\sum_{a = 1}^{i-1}P(a)
\end{bmatrix} 
, ~\text{  and  } ~Y(j):= \begin{bmatrix}
B_{1,j}\\
1\\
\sum_{a = 1}^{i-1}Q(a)
\end{bmatrix},
$$
we will have
$$\left|B_{i,j} - X(i)^T Y(j)\right| = \left|\sum_{a=1}^{i-1} \sum_{b=1}^{j-1} E_{a,b}\right| = O(Mn^{2t}),$$
which completes the induction. 
\end{proof}

\begin{theorem}\label{thm:DB_multiply}
For two integer matrices $A$ and $B$, if $| (\mathcal{D}^t B)_{i,j}| \le M$ for some constant $t \ge 1$, then there exists an algorithm that computes $A \star B$ in $\tilde{O}(n^{3-\frac{3-\omega}{2t^2+4}} M^{1/(2t^2+4)})$ time. 
\end{theorem}
\begin{proof}
Let $\Delta$ be a small polynomial in $n$. For any $\Delta \times \Delta$ sub-matrix of $B$, the $t$-th discrete difference is also bounded by $M$. Therefore, by Lemma \ref{lem:D^k}, for each $i', j'$ multiples of $\Delta$, there exist $2t$-dimensional vectors $X_{i',j'}(i), Y_{i',j'}(j)$ such that $X_{i',j'}(i)^T Y_{i',j'}(j)$ is an $O(M \Delta^{2t})$-additive approximation of $B_{i,j}$. In other word, every $\Delta \times \Delta$ sub-matrices of $B$ has an $O(M \Delta^{2t})$-approximate rank at most $2t$. Therefore, we can apply Theorem \ref{thm:low_rank} to get an algorithm that computes $A \star B$ in time
$$\tilde{O}(n^3  \cdot \Delta^{-1/\lfloor (2t+1)/2 \rfloor} + n^{(9+\omega)/4}\cdot (M \Delta^{2t})^{1/4}).$$
By setting $\Delta = \left(n^{(3-\omega)/2}\cdot M^{-1/2}\right)^{\frac{t}{t^2+2}}$, we get a $\tilde{O}(n^{3-\frac{3-\omega}{2t^2+4}} M^{1/(2t^2+4)})$ time algorithm. 
\end{proof}

\begin{corollary}
Given an $n \times n$ array $A$, where the absolute value of each entry is bounded by $M$. There exists an algorithm that finds the maximum subarray of $A$ in $\tilde{O}(n^{\frac{15+\omega}{6}} M^{1/6})$ time. Use $\omega < 2.373$, this gives an $\tilde{O}(n^{2.8955} M^{1/6})$ time algorithm, which is truly subcubic when $M=o(n^{0.627})$. 
\end{corollary}
\begin{proof}
We can use Tamaki and Tokuyama's reduction in \cite{tamaki1998algorithms}, and apply Theorem \ref{thm:DB_multiply} using $t=1$ to immediately get this result. 
\end{proof}

%% file: lower_bound.tex
In this section, we show  the conditional lower bound for the $d$-Dimensional Maximum Subarray problem, where the entries of the input array can have arbitrary real values. Backurs et al. \cite{backurs2016tight} showed an $n^{d+\lfloor d/2\rfloor -o(1)}$ conditional lower bound for $d$-Dimensional Maximum Subarray, based on the hardness of the Max-Weight $(d+\lfloor d/2\rfloor)$-Clique problem. Their lower bound is only tight for $d=2$, since Kadane's algorithm for $d$-Dimensional Maximum Subarray runs in $O(n^{2d-1})$ time. 

We show an $n^{2d-1-o(1)}$ conditional lower bound for the $d$-Dimensional Maximum Subarray problem, based on the hardness of the Max-Weight $(2d-1)$-Clique problem. In our reduction, we will introduce two intermediate problems defined as following. 

\begin{definition}[Two-sided $d$-Uniform Hypergraph]
A complete hyperedge-weighted $d$-uniform hypergraph whose vertex set is partitioned into $2d$ sets $U_1, U_2, \ldots, U_d,$ $V_1, V_2, \ldots, V_d$, each with $n$ vertices is {\em two-sided} if any $d$-hyperedge $(w_1,\ldots, w_d)$ not in the form of $w_1 \in U_1 \cup V_1, w_2 \in U_2 \cup V_2, \ldots, w_d \in U_d \cup V_d$, has zero weight. 
\end{definition}

\begin{definition}[Two-sided $d$-Uniform Max-Weight Hyperclique]
\label{dfn:two_sided_hyper}
Given a two-sided $d$-uniform hypergraph, find one vertex from each vertex set, so that the sum of hyperedge weights between these vertices is maximized. 
\end{definition}

\begin{definition}[Central $d$-Dimensional Array]
A $d$-dimensional array $A$ with side length $2n+1$ is called a central array if the index set of it is $\{-n, -n+1, \ldots, n-1, n\}^d$. 
\end{definition}

\begin{definition}[Central Maximum Subarray Sum]\label{dfn:central_subarray_prob}
Given a central $d$-dimensional array $A$, find 
$$\max_{\substack{i \in [n]^d \\ \delta \in [2n]^d\\ -n \le i-\delta < 0}} \sum_{j \in \{0,1\}^d} A[i - \delta \odot	 j],$$
where $\odot$ denotes the componentwise product of two vectors.
\end{definition}
Central Maximum Subarray Sum asks to find a subarray whose $2^d$ corners are in each of the $2^d$ quadrants, such that the sum of values on its corners is maximized. Backurs et al. \cite{backurs2016tight} showed an $O(n^d)$ time reduction from the Central Maximum Subarray Sum problem to the Maximum Subarray problem in $d$-dimension. Thus, any (higher than $n^d$) lower bound for the Central Maximum Subarray Sum problem would imply the same lower bound for the Maximum Subarray problem. In the rest of this section, we will first show a reduction from Max-Weight $(2d-1)$-Clique problem to Two-sided $d$-Uniform Max-Weight Hyperclique problem, and then show a reduction from the Two-sided $d$-Uniform Max-Weight Hyperclique problem to the Central Maximum Subarray Sum problem. If the well-known Max-Weight $(2d-1)$-Clique Hypothesis is true, the Central Maximum Subarray Sum problem would have an $n^{2d-1-o(1)}$ lower bound, and thus the Maximum Subarray problem would share the $n^{2d-1-o(1)}$ lower bound due to Backurs et al.'s reduction. 

\begin{lemma}\label{lem:clique_to_hyper}
If there exists an $O(n^{2d-1-\epsilon})$ time algorithm (for $\epsilon>0$) for the Two-sided $d$-Uniform Max-Weight Hyperclique problem, then there exists an $O(n^{2d-1-\epsilon})$ time algorithm for Max-Weight $(2d-1)$-Clique problem.
\end{lemma}
\begin{proof}
Let $G=(V_1 \cup V_2 \cup \cdots \cup V_{2d-1}, E)$ be a $(2d-1)$-partite graph. We will construct a Two-sided $d$-Uniform Hypergraph $G'=(U_1 \cup U_2 \cup \cdots \cup U_{2d}, E')$ such that the maximum $(2d-1)$-clique weight of $G$ is equal to the maximum $(2d)$-hyperclique weight of $G'$. For simplicity, assume $n$ is a power of $2$, and we will index the vertices in each vertex set from $0$. 

The first $2d-1$ vertex sets of $G'$ are copies of the vertex sets of $G$. Specifically, $U_i$ is a copy of $V_i$ for any $i \le 2d-1$. $U_{2d}$, however, encodes something different. Assume we pick $v_i \in V_i$ to be the $s_i$-th vertex in $V_i$, then intuitively, $U_{2d}$ encodes $s_{d+1} \oplus s_{d+2} \oplus \cdots \oplus s_{2d-1}$, where $\oplus$ is the bitwise exclusive-or operation. 

We initialize all hyperedge weights of $G'$ to $0$, and increase these weights incrementally by considering edges of $G$ one by one. 

For any $1 \le i < j \le 2d-1$, pick an edge $(v_i, v_j) \in V_i \times V_j$, with weight $w(v_i, v_j)$. Let $u_i, u_j$ be the copies of $v_i, v_j$ in the hypergraph $G'$. First consider the case when $j \ne i+d$. This is the case when there exist arbitrarily weighted hyperedges that contain both $u_i$ and $u_j$. Let $S := \{k \in [d]: k \not \equiv i \pmod{d} \text{ and } k \not \equiv j \pmod{d}\}$. We enumerate every $n^{d-2}$ combinations of vertices $u'_k \in U_k$ for $k \in S$, and add $w(v_i, v_j)$ to the hyperedge between the $d$ vertices $u_i, u_j$ and $u'_k$ where $k \in S$. 

The case when $j = i+d$ is more interesting, since all hyperedges  in $G'$ that contain both $u_i$ and $u_j$ must have zero weight, because of the definition of Two-sided $d$-Uniform Hypergraph. However, we can encode $w(v_i, v_j)$ via the extra vertex set $U_{2d}$. Let $u_j$ be the $s_j$-th vertex in $U_j$. We enumerate all $n^{d-2}$ combinations of indices $s'_{d+1}, s'_{d+2}, \ldots, s'_{j-1}, s'_{j+1}, \ldots, s'_{2d}$, such that $s'_{d+1} \oplus s'_{d+2} \oplus \cdots \oplus s'_{j-1} \oplus s_j \oplus s'_{j+1} \oplus \cdots \oplus s'_{2d-1} = s'_{2d}$. Let the $s'_k$-th vertex in $U_k$ be $u'_k$ for any $k \in \{d+1, d+2, \ldots, j-1,j+1, \ldots, 2d\}$. We add $w(v_i, v_j)$ to the hyperedge that consists of $u_i$ and $u'_k$ for every $k \in \{d+1, d+2, \ldots, j-1,j+1, \ldots, 2d\}$. 

Finally, enumerate all combinations of $s_{d+1}, s_{d+2}, \ldots, s_{2d}$ such that $s_{d+1} \oplus s_{d+2} \oplus \cdots s_{2d-1} \ne s_{2d}$. Let $u_k$ be the $s_k$-th vertex in $U_k$, for every $d+1 \le k \le 2d$.  We set the weight of the hyperedge that consists of $u_{d+1}, u_{d+2}, \ldots, u_{2d}$ to $-M'$ for some large enough $M'$. If all edge weights in $G$ are numbers in $[-M, M]$, we can set $M'$ to be $100d^{10}M$. 

The construction of $G'$ takes $O(n^d)$ time, since for each edge $(v_i, v_j)$ in $G$, we enumerate $O(n^{d-2})$ hyperedges. It remains to show that the maximum weight of $(2d-1)$-cliques in $G$ is equal to the maximum $(2d)$-hyperclique weight of $G'$. 

Pick any $2d$ indices $s_1, s_2, \ldots, s_{2d}$. Let $u_i$ be the $s_i$-th vertex in $U_i$. If $s_{d+1} \oplus s_{d+2} \oplus \cdots \oplus s_{2d-1} \ne s_{2d}$, then there will be a $-M'$ weight on the hyperedge $(u_{d+1}, u_{d+2}, \ldots, u_{2d})$, so the weight of the hyperclique $u_1, u_2, \ldots, u_{2d}$ can never be maximum. Therefore, we are forced to pick $s_{2d} = s_{d+1} \oplus s_{d+2} \oplus \cdots \oplus s_{2d-1}$. In this case, the weight of the hyperclique $(u_1, u_2, \ldots, u_{2d})$ is equal to the weight of the clique $(v_1, v_2, \ldots, v_{2d-1})$, where $v_i$ is a copy of $u_i$ for each $i < 2d$. Thus, if we invoke the $O(n^{2d-1-\epsilon})$ algorithm for the Two-sided $d$-Uniform Max-Weight Hyperclique problem on $G'$, we get the Max-Weight $(2d-1)$-Clique on $G$. 
\end{proof}

\begin{lemma}\label{lem:hyper_to_central}
If there exists an $O(n^{2d-1-\epsilon})$ time algorithm for the $d$-Dimensional Central Maximum Subarray Sum problem, then there exists an $O(n^{2d-1-\epsilon})$ time algorithm for the Two-sided $d$-Uniform Max-Weight Hyperclique problem. 
\end{lemma}
\begin{proof}
Take any Two-sided $d$-Uniform Hypergraph $G=(V_1 \cup V_2 \cdots V_d \cup U_1 \cup U_2 \cdots \cup U_d, E)$, we index the vertices in $V_i$ and $U_i$ from $1$. We will construct a $d$-Dimensional Central Array $A$ based on $G$ such that the central maximum subarray sum of $A$ is equal to the maximum $2d$-hyperclique of $G$. If any entry of vector $i \in \{-n, \ldots, n\}^d$ is $0$, we set $A[i]$ to be $0$, since they are not relevant to the central maximum subarray sum of $A$. For any other index $i$, we choose $d$ vertices $w_1, w_2, \ldots, w_d$ from  the graph $G$ based on $i$. If $i_r > 0$ for some $r$, we choose $w_r$ to be the $i_r$-th vertex in $V_r$; otherwise, we choose $w_r$ to be the $(-i_r)$-th vertex in $U_r$. We set $A[r]$ to be the weight of the hyperedge connecting $w_1, w_2, \ldots, w_d$. 
 
Pick any $2d$ vertices $v_1 \in V_1, v_2 \in V_2, \ldots, v_d \in V_d, u_1 \in U_1, u_2 \in U_2 \ldots, u_d \in U_d$. Define two $d$-dimensional vectors $\vec{v}$ and $\vec{u}$, such that $\vec{v}_i$ is the index of $v_i$ in $V_i$, and $\vec{u}_i$ is the index of $u_i$ in $V_i$. For any $j \in \{0, 1\}^d$, let $i$ be a $d$-dimensional vector such that $i_r = \vec{v}_r$ if $j_r = 0$, and $i_r = -\vec{u}_r$ if $j_r = 1$. Also, let $W$ be a set of $d$ vertices $\bigcup_{1 \le r \le d}\{\text{if } j_r = 0 \text{ then } v_r \text{ else } u_r\}$. The entry $A[i]$ is exactly the weight of the hyperedge between vertices in $W$. Thus, the sum of all corners of the subarray whose two opposite corners are $\vec{v}$ and $-\vec{u}$ is equal to the weight of the hyperclique $(v_1, v_2, \ldots, v_d, u_1, u_2, \ldots u_d)$. 

Conversely, for any subarray of $A$ whose $2^d$ corners are in different quadrants, there exists a hyperclique in $G$ whose $2d$ vertices are from different vertex sets, by a similar argument. 

Thus, the central maximum subarray sum of $A$ is equal to the max-weight $2d$-hyperclique of $G$, so we can invoke the $O(n^{2d-1-\epsilon})$ algorithm of $d$-Dimensional Central Maximum Subarray Sum problem to output the max-weight $2d$-hyperclique of $G$. 
\end{proof}

Lemma \ref{lem:clique_to_hyper}, Lemma \ref{lem:hyper_to_central}, together with the reduction from Central Maximum Subarray Sum to Maximum Subarray \cite{backurs2016tight} imply Theorem \ref{thm:subarray_lowerbound}.

%% file: lower_bound_bounded.tex
In Section \ref{sec:lower_bound}, we showed a tight conditional lower bound for $d$-Dimensional Maximum Subarray with real valued weights. In Section \ref{sec:upper_bound}, we also showed an algorithm that is better than this conditional lower bound, for 2D Maximum Subarray with bounded integer weights. A natural question arises: Can we prove some conditional lower bound for $d$-Dimensional Maximum Subarray when the numbers in the array are bounded integers? 

In this section, we answer this question positively by proving Theorem \ref{thm:bounded_subarray_lowerbound}. 
We notice that the reduction from Two-sided $d$-Uniform Max-Weight Hyperclique problem to Central Maximum Subarray Sum (Lemma \ref{lem:hyper_to_central}), and the reduction from Central Maximum Subarray Sum to Maximum Subarray (presented in \cite{backurs2016tight}) only increase the largest absolute value of weights by a constant factor. Therefore, we only need to show a conditional lower bound for Two-sided $d$-Uniform Max-Weight Hyperclique when the weights of the hyperedges are bounded integers. Therefore, Theorem \ref{thm:bounded_subarray_lowerbound} follows from the following lemma.  

\begin{lemma}\label{lem:clique_to_hyper_const}
If there exists an $O(n^{2d-2-\epsilon})$ algorithm (for $\epsilon > 0$) for the Two-sided $d$-Uniform Max-Weight Hyperclique problem where the hyperedges have bounded integer weights, then there exists an $O(n^{2d-2-\epsilon})$ algorithm for the $3$-Uniform $(2d-2)$-Hyperclique problem. 
\end{lemma}
\begin{proof}
The proof has a similar spirit as the proof to Lemma \ref{lem:clique_to_hyper}. For simplicity, we denote a $3$-Uniform Hypergraph $G$ as $(V_1 \cup V_2 \cup \ldots \cup V_{d-1} \cup U_1 \cup U_2\cup \cdots \cup U_{d-1}, E)$. Note that even though the vertex sets of $G$ are not partitioned to two parts naturally, we used $V_i$ for one half and $U_i$ for the other half. Also assume $n$ is a power of $2$ for simplicity. 

Create a two-sided $d$-uniform hypergraph $G' = (V_1' \cup V_2' \cup \cdots \cup V_d' \cup U_1' \cup U_2' \cup \cdots \cup U_d', E')$, where $V_i'$ is a copy of $V_i$ for any $i \le d-1$, and $U_i'$ is a copy of $U_i$ for any $i \le d - 1$. If we pick the $s_i$-th vertex $v_i'$ from $V_i'$, then $V_d'$ encodes the information $s_1 \oplus s_2 \oplus \cdots \oplus s_{d-1}$. Similarly, if we pick the $t_i$-th vertex $u_i'$ from $U_i'$, then $U_d'$ encodes the information $t_1 \oplus t_2 \oplus \cdots \oplus t_{d-1}$. We call $V_i$ and $U_i$ corresponding vertex sets; we also call $V_i'$ and $U_i'$ corresponding vertex sets. For any vertex set $S$, we use $S[k]$ to denote the $k$-th vertex in $S$, indexed from $0$. 
We initialize all edge weights of $G'$ to $0$.

Every $3$-hyperedge $(a, b, c) \in E$ will increase the weight of some hyperedges in $G'$ by $1$. First assume no two vertices in $\{a, b, c\}$ are from a pair of corresponding vertex sets. Let $a', b', c'$ be $a,b,c$'s copies in $G'$, respectively. Take any hyperedege $e'$ in $G'$ that contains $\{a', b', c'\}$ and $d-3$ other vertices from the first half of the vertex sets, so that every pair of corresponding vertex sets contains exactly one vertex. We increment the weight of any such hyperedge $e'$ by $1$. 

If two vertices in $\{a, b, c\}$ are from a pair of corresponding vertex sets, then without loss of generality, we can assume $a \in V_i, b \in U_i$. When $c \in V_j$ for some $j \ne i$, we increment all hyperedges consisting of vertices $V_1'[s_1], \ldots, V_{i-1}'[s_{i-1}], U'_{i}[t_i], V_{i+1}'[s_{i+1}], \ldots, V_d'[s_d]$, where $U_{i}[t_i] = b$, $V_{j}[s_j] = c$ and $V_{i}\left[\bigoplus_{1 \le k \le d, k \ne i} s_k\right] = a$. It is symmetric when $c \in U_j$ for some $j \ne i$: we increment all hyperedges consisting of vertices $U_1'[t_1], \ldots, U_{i-1}'[t_{i-1}], V'_{i}[s_i], U'_{i+1}[t_{i+1}], \ldots, U_d'[t_d]$, where $V_{i}[s_i] = a$,  $ U_{j}[t_j] = c$ and $U_{i}\left[\bigoplus_{1 \le k \le d, k \ne i} t_k\right] = b$. 

Finally, for any $s_1, s_2, \ldots, s_d$ such that $s_1 \oplus s_2 \oplus \cdots \oplus s_{d-1} \ne s_d$, we set the weight of the edge consisting of $V_1'[s_1], V_2'[s_2], \ldots, V'_d[s_d]$ to be $-M$ for $M = 100d^{10}$. Symmetrically, for any $t_1, t_2
, \ldots, t_d$ such that $t_1 \oplus t_2 \oplus \cdots \oplus t_{d-1} \ne t_d$, we set the weight of the edge consisting of $U_1'[t_1], U_2'[t_2], \ldots, U_d'[t_d]$ to be $-M$. 

The maximum absolute value of hyperedge weight of $G'$ is $M$, which is a constant when $d$ is a constant. 
By construction, $G$ has a $(2d-2)$-hyperclique if and only if the max-weight hyperclique of $G'$ has weight $\binom{2d-2}{3}$. Thus, we can solve  $3$-Uniform $(2d-2)$-Hyperclique by invoking the assumed algorithm for the Two-sided $d$-Uniform Max-Weight Hyperclique problem.

\end{proof}

%% file: appendix.tex
\section{Derandomization of the Main Algorithm}\label{sec:general_derandom}
The only randomness used by the algorithm is to sample random $j^r \in [n]$ in Phase $2$. In order to remove this randomness, we need to first define the following notion of \textit{approximately relevant triples}. 

\begin{definition}
A triple $(i, k, j)$, where $k \in I(k'), j \in I(j')$ for some $k', j'$ divisible by $\Delta$, is called \textit{approximately relevant} if $\left|A_{i,k} + X_{k',j'}(k)^T Y_{k',j'}(j) - \tilde{C}_{i,j}\right| \le 4W$. 
\end{definition}

Approximately relevant triples are strongly related to weakly relevant triples by the following lemma. 

\begin{lemma}\label{lem:weak_is_approx}
Any weakly relevant triple $(i, k, j)$ is also approximately relevant. 
\end{lemma}
\begin{proof}
Consider
\begin{equation*}
    \begin{split}
        &\left|\left(A_{i,k} + X_{k',j'}(k)^T Y_{k',j'}(j) - \tilde{C}_{i,j}\right) - \left(A_{i,k}+B_{k,j}-\tilde{C}_{i,j} \right)\right|\\
        =&\left|X_{k',j'}(k)^T Y_{k',j'}(j) - B_{k,j}\right| \le W
    \end{split}
\end{equation*}
Therefore, by the simple inequality $||a|-|b|| \le |a-b|$, we know that $$\left|\left|A_{i,k} + X_{k',j'}(k)^T Y_{k',j'}(j) - \tilde{C}_{i,j}\right| - \left|A_{i,k} + B_{k,j} - \tilde{C}_{i,j}\right|  \right| \le W.$$
For any weakly relevant triple $(i, k, j)$, $\left|A_{i,k}+B_{k,j}-\tilde{C}_{i,j}\right| \le 3W$ by definition. Since the difference between it and $\left|A_{i,k} + X_{k',j'}(k)^T Y_{k',j'}(j) - \tilde{C}_{i,j}\right|$ is bounded by $W$,  the latter cannot exceed $4W$, which means $(i, k, j)$ is approximately relevant. 
\end{proof}

Therefore, in order to cover approximately relevant triples, when computing $A^r \star B^r$, we need to keep all entries of $A^r$ that have absolute value at most $5W$, but it won't change time complexity. 

After we sample some $j^r$, if the number of uncovered, approximately relevant triples is $O(n^3 /\rho)$, then by Lemma \ref{lem:weak_is_approx}, the number of uncovered, weakly relevant triples is $O(n^3 /\rho)$ as well. In the rest of this section, we show how to \textit{deterministically} choose the set of $j^r$, so that the number of uncovered, approximately relevant triples is $O(n^3 /\rho)$ after computing $A^r \star B^r$ for all $j^r$. 

We first notice that a triple $(i, k, j)$ is approximately relevant if and only if $A_{i,k} + X_{k',j'}(k)^T Y_{k',j'}(j) - \tilde{C}_{i,j} \le 4W$, since this quantity can never be smaller than $-4W$. 
Fix $i', k', j', i \in I(i')$. For every $j \in I(j')$, we add the point $\begin{bmatrix}
-\tilde{C}_{i,j}\\
Y_{k',j'}(j)
\end{bmatrix}$ to the geometric data structure. This takes $\tilde{O}(n^3/\Delta)$ time. Then for each $k \in I(k')$, we use the geometric data structure to list points in the half-space
$\begin{bmatrix}
-A_{i,k}\\
X_{k',j'}(k)
\end{bmatrix}^T x \le 4W$. It will take $\tilde{O}(n^3 \cdot \Delta^{-1/\lfloor (d+1)/2\rfloor}) + O(\text{total number of points listed})$. For each $(i, k)$ pair, if we stop listing $j$ as soon as we get $n/\rho$ values of $j$, the total number of points listed would be $O(n^3 /\rho)$. 

Finally, for the $(i, k)$ pairs that have less than $n/\rho$ values of $j$ listed, we ignore these pairs . For every other pair $(i, k)$, we have a set $S(i, k)$ that contains $n/\rho$ values of $j$ such that $(i, k, j)$ is approximately relevant. We need to find a set of $j^r$ that intersects with each of these $S(i, k)$ sets. By the standard greedy algorithm for hitting set/set cover, we can choose $\tilde{O}(\rho)$ different $j^r$ in $\tilde{O}(n^3/\rho)$ time, so that each $S(i, k)$ contains at least one $j^r$ we choose. 

The other parts of the algorithm proceed similarly, and it will have the same running time as the \textit{randomized} version. 

\section{Limitation of the Reduction Path for Constant Weight Maximum Subarray}

Our conditional lower bound in Section \ref{sec:lower_bound_bounded} first reduces a hardness problem to the Central Maximum Subarray Sum problem, and then to the Maximum Subarray problem. Backurs et al. \cite{backurs2016tight} use a similar strategy in their reduction. Vassilevska W. and Williams \cite{vw10j} also have an intermediate problem in their reduction from Negative Triangle to 2D Maximum Subarray. This intermediate problem, similar to the Central Maximum Subarray Sum problem, also weights a subarray based on the values on the corner of the subarray.

The Central Maximum Subarray Sum problem, though nicely fits in all these previous reductions to Maximum Subarray, has a major limitation as the intermidiate problem: if the weights of the array are bounded integers, then there exists an $\tilde{O}(n^{2d-4+\omega})$ time algorithm that solves the Central Maximum Subarray Sum problem. It means that, in order to prove an lower bound larger than $n^{2d-4+\omega}$ for Maximum Subarray, we need to find some other alternative intermediate problem. 

\begin{claim}\label{lem:faster_corner}
Given a $d$-Dimensional Central Array $A$, such that all entries of $A$ are integers bounded by some constant, there exists a $\tilde{O}(n^{2d-4+\omega})$ time algorithm that computes Central Maximum Subarray Sum of $A$.
\end{claim}
\begin{proof}[Proof Sketch]
When $d>2$, we can exhaustively enumerate all possible values of the first $d-2$ dimensions, and weights of the remaining $2$D problem can be at most $2^{d-2}$ times larger than the original weights. When $d$ is a constant, the resulting $2$D problem also has entries bounded by a constant. Therefore, it is sufficient to show an $\tilde{O}(n^\omega)$ algorithm for the $2$D case. 

The 2D case is similar to Tamaki et al.'s algorithm for 2D Maximum Subarray \cite{tamaki1998algorithms}. Using Min-Plus product for matrices with bounded integer weights, we can compute in $\tilde{O}(n^\omega)$ time, for every $x_1 < 0 < x_2$, 
$$1)\ d_{x_1, x_2} = \max_{y_1 < 0} \left\{A_{x_1, y_1} + A_{x_2, y_1}\right\}, \ \ \ \ \ \ \ \ \ 2)\  u_{x_1, x_2} = \max_{y_2 > 0} \left\{ A_{x_1, y_2} + A_{x_2, y_2}\right\}. $$
Then the central maximum subarray sum of $A$ is $\max_{x_1 < 0 < x_2} \left( d_{x_1, x_2} + u_{x_1, x_2} \right)$.  
\end{proof}